\documentclass[a4paper,10pt,english]{article}
\usepackage[utf8]{inputenc}
\usepackage[T1]{fontenc}
\usepackage{authblk}
\usepackage{amsmath,amssymb,amsthm,bm}
\usepackage{longtable}
\usepackage{array}
\usepackage{graphicx}
\usepackage[export]{adjustbox}
\usepackage{lipsum}
\usepackage{enumitem}
\usepackage{subcaption}
\usepackage{float}
\usepackage[]{sidecap}
\usepackage{wrapfig}
\usepackage[colorlinks=true, linkcolor=red, urlcolor=red, citecolor=red]{hyperref}
\usepackage{comment}
\usepackage{braket}
\usepackage[table]{xcolor}
\usepackage[all]{xy}
\usepackage{tikz,tikz-cd,tikz-3dplot}
\usepackage{quantikz}
\usepackage{xcolor,soul,tcolorbox}

\newtheorem{definition}{Definition}
\newtheorem{proposition}{Proposition}

\newtheorem{corollary}{Corollary}

\setlist[itemize]{itemsep=0pt, topsep=2pt, parsep=0pt, partopsep=2pt}
\setlist[enumerate]{itemsep=0pt, topsep=2pt, parsep=0pt, partopsep=2pt}

\title{SBN Explorer\\ \vspace*{2mm} \large An Empirical Study of Cryptographic Boolean Networks}
\author{Arnaud Valence}
\date{}

\begin{document}
	\maketitle
	\tableofcontents

\begin{abstract}
	Boolean circuits form the foundational computational substrate of symmetric cryptography, yet the exploration of their architectural design space has remained largely confined to a handful of canonical paradigms --- SPN, Feistel networks, and their immediate variants. This paper takes a deliberately broader perspective by formalizing the design space of cryptographic Boolean systems through six independent binary structural constraints: Stratification, Acyclicity, Regularity, Interleaving, Homogeneity, and Locality. These constraints generate a hypercube of $2^6 = 64$ distinct architectural classes defined over Synchronous Boolean Networks, a general model that subsumes both acyclic combinational circuits and recurrent synchronous systems. We systematically evaluate all 64 classes against three generic cryptanalytic fitness objectives --- differential, linear and algebraic resistance --- using a five-stage methodology centered on Formal Concept Analysis. The results reveal that the best Boolean networks are governed by the identification of sparse, mutually compatible combinations of constraints --- a fundamentally epistatic problem that classical cryptography has barely addressed.
	
\end{abstract}

\section{Introduction}  

Boolean circuits (BCs) constitute a central computational model in symmetric cryptography, where security ultimately emerges from the coordinated action of linear diffusion and nonlinear confusion mechanisms. In the standard design literature, this interaction is most often realized through highly structured architectures, the canonical example being the Substitution--Permutation Network (SPN), which alternates nonlinear substitution layers with linear diffusion layers. Although this paradigm has proved remarkably successful, it also tends to narrow the architectural space explored in practice, by implicitly treating one particular organizational principle as a quasi-default design pattern.

This work adopts a deliberately more general perspective, independent of any specific architecture. In particular, the principles of acyclicity and alternation of linear and nonlinear operators are no longer considered fundamental requirements, but rather possible architectural properties that can be enabled or disabled.

More precisely, our framework is based on six independent binary structural constraints, generating a design space of $2^6 = 64$ distinct architectural classes. This combinatorial approach allows us to move beyond isolated, manually designed topologies and to study, within a unified formal framework, how different structural prescriptions influence the cryptographic potential of BCs.

The motivation of the paper is therefore to systematically explore alternative design strategies for cryptographic BCs. By enlarging the admissible architectural space beyond conventional acyclic and SPN-centered constructions, we aim to identify whether non-standard structural organizations can also support strong cryptographic behavior, and more generally to clarify which architectural principles are genuinely beneficial for the design of robust cryptographic Boolean systems.

\paragraph{Limitations of Existing Approaches.}

Existing approaches to the design of cryptographic BCs remain largely dominated by a small number of canonical architectural paradigms, such as SPNs, Feistel networks, and closely related round-based architectures built from successive nonlinear and linear transformation stages. This has produced highly successful designs, but it has also constrained the exploration of the structural design space. In most cases, the architecture is fixed \emph{a priori}, and the analysis focuses on optimizing components inside an already accepted organizational scheme rather than questioning the scheme itself.

A first limitation is the strong implicit bias toward acyclic and strictly layered models. This preference is justified not only by formal considerations -- combinational circuits are traditionally defined as memoryless input-output applications -- but also by practical material constraints -- acyclic structures are significantly easier to synthesize, verify, and optimize. Consequently, much of the literature considers acyclicity an essential condition for circuit correctness, which immediately excludes a broader class of synchronous Boolean systems with feedback. However, potentially relevant non-standard architectures are not compared on an equal footing with conventional ones, but are discarded at the modeling stage.

A second limitation lies in the privileged status granted to the classical alternation between nonlinear and linear operators. This principle, historically linked to Shannon's insights \cite{shannonCommunicationTheorySecrecy1949a}, is often treated as an almost necessary design rule rather than as one structural option among others. Yet from a more abstract viewpoint, there is no reason to assume in advance that this specific arrangement exhausts the space of architectures capable of supporting strong cryptographic properties. Other structural principles may also influence resistance to differential, linear, or algebraic attacks, either positively or negatively, and their role remains insufficiently isolated.

A third limitation is methodological. Existing studies typically compare a small number of hand-designed architectures, often motivated by implementation practice or historical design traditions. Such comparisons are informative, but they do not provide a systematic exploration of the architectural hypercube generated by multiple independent structural constraints. In particular, they make it difficult to determine whether an observed cryptographic advantage is due to one constraint, to a combination of constraints, or simply to incidental design choices specific to a given family.

Finally, many existing analyses emphasize component-level properties -- for example, the quality of S-boxes, linear layers, or local propagation parameters -- without jointly addressing the global topology of the circuit as an explicit object of study. This creates a gap between local cryptographic criteria and global structural organization. From our perspective, what is still missing is a framework in which architectural constraints can be formalized independently, combined systematically, and evaluated quantitatively across a broad family of BCs.

These limitations motivate the present work: rather than starting from a predefined cryptographic architecture and refining it locally, we study the architecture itself as a variable, and we do so in a sufficiently general setting to include both standard BCs and more general synchronous Boolean networks.

\paragraph{Research Question and Contributions.}

The central question of this paper is whether the empirically measured cryptographic quality of a BC is intrinsically tied to a small set of conventional architectural principles, or whether comparably strong -- or possibly stronger -- properties can emerge from alternative structural organizations. Stated differently, we ask to what extent the global topology of a cryptographic Boolean system acts as a determining factor in its resistance to major forms of cryptanalysis.

To address this question, we start from a general class of Boolean functions and define six independent binary structural constraints capturing distinct architectural principles. This yields a finite design space of $64$ architecture classes, each corresponding to a specific combination of constraints. Within this framework, the precise research question can be formulated as follows: \emph{how do different combinations of independent structural constraints affect the cryptographic robustness of Boolean systems, when robustness is evaluated through generic criteria related to resistance against differential, linear, and algebraic cryptanalysis?}

In this sense, the paper is concerned with a general problem of architectural cryptography: determining how much of the security of a Boolean construction comes from its local components, and how much comes from the global organization of the circuit itself. The main contributions of this paper are the following.

First, we introduce a general architectural framework for cryptographic Boolean systems based on six independent binary structural constraints, defined over the broad class of Synchronous Boolean Networks (SBNs).

Second, we perform a systematic exploration of the resulting $64$ architecture classes, in order to compare their cryptographic behavior with respect to generic indicators related to differential, linear, and algebraic resistance.

Third, we investigate whether the space of high-performance architectures allows us to identify closed formal concepts of cryptographic robustness -- that is, recurring and stable combinations of architectural properties.

Finally, this article does not simply address the qualitative aspects of cryptographic robustness, i.e., formal concepts; it also highlights the quantitative contributions of relevant cryptographic concepts through the spectral variance decomposition of the hypercube of cryptographic architectures.

\paragraph{Roadmap.}

The remainder of the paper is organized as follows. Section~2 reviews the existing literature on cryptographic BCs and non-standard structural approaches. Section~3 introduces the formal framework of the study and recalls the main definitions used throughout the paper, including BCs, SBNs, architectural constraints, and cryptographic evaluation criteria. Section~4 presents the methodology, namely the statistical processing pipeline --- first qualitative, then quantitative --- used to derive cryptographic robustness results. Section~5 reports the main results of the systematic exploration and comparative analysis of the $64$ architecture classes. Section~6 discusses the structural meaning of these results and puts them into perspective with respect to standard cryptographic BCs. Finally, Section~7 concludes the paper and outlines several directions for future research. The appendix also includes certain formal and mathematical details of our analysis, as well as a technical description of the “SBN Explorer” program, which is available on GitHub.

\section{Related Work}  

\subsection{BCs in Cryptographic Design}

BCs constitute the fundamental computational substrate of symmetric cryptography. At a sufficiently low level of abstraction, any block cipher, stream cipher, or permutation-based primitive can be modeled as a BC implementing a Boolean function $F : \{0,1\}^n \rightarrow \{0,1\}^n$. Accordingly, cryptographic BCs have emerged as a central domain within BC theory and are now supported by a mature and substantial research literature \cite{wuBooleanFunctionsTheir2016,cusickCryptographicBooleanFunctions2017,carletBooleanFunctionsCryptography2020}.

This circuit-level view is independent of higher-level design paradigms and corresponds directly to the physical implementation of cryptographic algorithms in hardware or bit-sliced software. Consequently, BCs provide a natural framework for studying cryptographic transformations without assuming a particular architectural template.

\paragraph{Classical structured constructions.}

Despite this general nature, most modern symmetric primitives employ highly structured circuit topologies. Two dominant design paradigms illustrate this trend.

\begin{itemize}
	\item \textbf{Feistel networks.} The circuit is partitioned into two halves updated iteratively through round functions and XOR mixing. 
	\item \textbf{Substitution–Permutation Networks (SPN).} Functions are decomposed into alternating layers of linear transformations (diffusion layers through permutation boxes) and nonlinear transformations (confusion layers through substitution boxes). 
\end{itemize}

These architectures impose strong structural constraints on the underlying BC. The circuit is typically layered, acyclic, and organized into homogeneous rounds. Such regularity facilitates analysis, modular reasoning, and provable bounds on differential or linear propagation.

\paragraph{Circuit-level viewpoint versus architectural templates.}

From a purely Boolean perspective, however, these constructions represent only a very small subset of the space of valid circuits. A general BC may combine nonlinear and linear operations arbitrarily, exhibit irregular connectivity patterns, or distribute diffusion and confusion across the entire graph rather than separating them into explicit layers.

Historically, exploring this broader design space has proven difficult for two reasons:

\begin{itemize}
	\item \textbf{Analytical tractability.}  
	Most classical security arguments rely on structural regularities (round structure, S-box decomposition, independence assumptions).  
	When these assumptions are removed, the standard proof techniques become inapplicable.
	\item \textbf{Design methodology.}  
	Human-designed primitives typically rely on modular building blocks that simplify reasoning and implementation.  
	Unstructured circuits are difficult to design manually while maintaining cryptographic guarantees.
\end{itemize}

As a consequence, cryptographic research has traditionally focused on families of circuits exhibiting strong architectural regularity.

\paragraph{Emerging exploration of unconventional circuit topologies.}

Recent developments in automated design --- such as evolutionary search, reinforcement learning, and SAT-based optimization --- have started to explore the space of BCs more directly. These approaches often operate on circuit representations without imposing a predefined SPN or Feistel template.

Empirically, such methods tend to generate circuits with unconventional structural properties, including:

\begin{itemize}
	\item interleaving of nonlinear and linear operations within the same subgraph,
	\item dense cross-dependencies between input and output bits,
	\item heterogeneous logical depths across the circuit,
	\item compact blocks where diffusion and nonlinearity are internally entangled,
	\item or even feedback structures when the model allows synchronous updates.
\end{itemize}

These architectures do not conform to the classical layered view of symmetric cryptographic design. Nevertheless, they remain perfectly valid BCs implementing deterministic transformations over $\{0,1\}^n$.

\paragraph{Towards non-standard BC topologies.}

The existence of such architectures raises a fundamental question:  
to what extent are the structural constraints of classical designs intrinsic to cryptographic security, and to what extent are they merely methodological choices that simplify analysis?

Investigating this question requires considering BCs beyond the traditional SPN or Feistel frameworks. In particular, one may allow:

\begin{itemize}
	\item non-layered combinations of linear and nonlinear operations,
	\item cross-block dependencies and irregular connectivity patterns,
	\item conditional wiring structures such as multiplexed permutations,
	\item or even controlled feedback leading to synchronous dynamics.
\end{itemize}

The following section examines several categories of such non-standard circuit topologies and discusses their structural properties and potential cryptographic implications.

\subsection{Non-Standard Topologies and Emerging Structures in Cryptographic BCs}

From a broader circuit-design perspective, symmetric cryptography has long relied on a richer set of Boolean dependency graphs than the canonical SPN and balanced Feistel templates. Generalized Feistel schemes, introduced in the provable-security line of Zheng, Matsumoto, and Imai and later refined by Suzaki and Minematsu, showed that increasing the number of branches and carefully choosing the inter-branch permutation could improve diffusion and security-efficiency trade-offs beyond the classical two-branch Feistel setting \cite{ZhengMatsumotoImai1990, SuzakiMinematsu2010}. Lai-Massey constructions, first proposed by Lai and Massey and later recast by Vaudenay and by Yun, Park, and Lee within the broader quasi-Feistel framework, provided an alternative invertible round organization in which Feistel-style security arguments can be transferred while preserving a different symmetry of data mixing \cite{LaiMassey1991,Vaudenay1999,YunParkLee2011}. ARX designs pushed this departure further by replacing S-box-centered nonlinearity with modular addition, rotations, and XOR; in particular, Dinu et al. showed with SPARX that such word-oriented constructions can still admit explicit differential and linear security arguments through the long-trail strategy \cite{DinuEtAl2016}.

In stream-cipher design, feedback-based architectures such as Grain and Trivium rely on recurrent shift-register topologies rather than acyclic round functions, showing that cryptographic strength may arise from sparse local feedback and long state evolution rather than from homogeneous layered circuits \cite{HellJohanssonMeier2007, DeCannierePreneel2008}. Cellular-automata-based generators pursue a related logic at the level of local update rules; already Wolfram emphasized that simple one-dimensional cellular automata could serve as keystream generators, opening a line of research in which global complexity emerges from uniform neighborhood interactions \cite{Wolfram1986}.

At the implementation level, LS-designs made compact bitslice structures a deliberate design objective by combining regular Boolean circuitry with efficient masked software realizations \cite{GrossoEtAl2014}. More recently, automated frameworks such as CLAASP have provided a unified representation and analysis environment for heterogeneous symmetric constructions, which is especially relevant for systematic comparisons of unconventional circuit topologies \cite{BelliniEtAl2024}. Taken together, these lines of work show that the history of symmetric design already includes many families that relax at least one canonical assumption of standard block-cipher design, whether explicit separation of confusion and diffusion, strict layer homogeneity, round-wise modularity, or acyclicity of the dependency graph.

\paragraph{From canonical round structures to broader circuit families.}

Classical SPNs enforce a clean alternation between nonlinear substitution layers and linear diffusion layers. This decomposition is analytically convenient because it supports round-based arguments such as active S-box counting, branch-number reasoning, and Wide-Trail-style proofs \cite{DaemenRijmen2001}. Feistel networks retain a similarly regular round structure while relaxing the requirement that the round function itself be invertible, a property that helps explain their central role in provable block-cipher design \cite{LubyRackoff1988}. Yet several important families already depart from this strict organization. Generalized Feistel Networks (GFNs) distribute the state over more than two branches and allow richer inter-branch coupling patterns than the classical balanced Feistel construction; this broadens the space of admissible dependency graphs and has led to refined design strategies for improving diffusion across branches \cite{Nyberg1996, SuzakiMinematsu2010}. Likewise, Lai-Massey schemes occupy an intermediate position between the SPN and Feistel paradigms: they rely on a dedicated symmetric mixing operation between state halves, yielding an invertible round structure that is neither purely substitution-permutation nor purely Feistel \cite{LaiMassey1990,Vaudenay1999,YunParkLee2011}. These families show that non-standard topology is not external to mainstream cryptography; it is already part of the established design landscape.

\paragraph{Hybrid and heterogeneous round organizations.}

A second line of work concerns block ciphers whose internal organization is deliberately heterogeneous across rounds or across successive phases of encryption. In such designs, the overall transformation is not obtained by iterating a single homogeneous round module, but by composing structurally distinct stages that play different cryptographic roles. A classical example is MARS, whose encryption process combines input and output whitening with unkeyed mixing phases surrounding a keyed cryptographic core. As a result, its structure is not well captured by the repeated motif
\[
\text{nonlinear layer} \;\rightarrow\; \text{linear layer} \;\rightarrow\; \text{key addition},
\]
but instead by a succession of functionally differentiated layers. From the viewpoint of Boolean-circuit topology, this already constitutes a clear departure from uniform round-based architectures such as canonical SPNs. This observation is important for the study of learned or evolved circuits: once heterogeneous round roles are accepted in manually designed ciphers, the emergence of irregular phase structures in automated constructions becomes much less surprising.

\paragraph{Interleaving of confusion and diffusion.}

A major departure from classical SPNs is the disappearance of a clean boundary between nonlinear and linear components. This feature is particularly visible in ARX constructions, where modular addition, XOR, and rotation are combined at word level \cite{DinuEtAl2016}. In such designs, nonlinearity is not isolated inside explicit S-boxes; rather, it arises from the carry propagation of modular addition and from the repeated interaction between arithmetic and bitwise operations. Diffusion, likewise, is not concentrated in a single permutation or matrix layer, but emerges progressively through repeated local mixing across words and bit positions \cite{DinuEtAl2016}. 

As a consequence, ARX-based cryptographic circuits typically exhibit several non-standard structural features:
\begin{itemize}
	\item no identifiable S-box layer;
	\item no explicit P-box or MDS layer in the classical SPN sense;
	\item nonlinearity distributed across the datapath rather than localized in a small family of boxes;
	\item security analyses based on ARX-specific techniques, such as dedicated differential and linear trail search or rotational cryptanalysis, rather than on active-S-box counting.
\end{itemize}

From a circuit perspective, ARX constructions therefore provide a canonical example of \emph{interleaved circuits}, in which confusion and diffusion are entangled within the same local operations \cite{BiryukovEtAl2016,KhovratovichNikolic2010}.

\paragraph{Triangular and data-dependent structures.}

Another important departure from standard round-based architectures is provided by triangular mappings, in particular T-functions. A T-function is a mapping in which the $i$-th output bit depends only on input bits of index at most $i$, under a fixed bit ordering \cite{KlimovShamir2004,KlimovShamir2005}. This induces a strongly asymmetric dependency graph: instead of balanced mixing across coordinates, information propagates according to a triangular causal order. From the viewpoint of classical block-cipher design, such a topology is neither layered in the SPN sense nor branch-based in the Feistel sense. Its cryptographic interest stems from the fact that many machine-oriented operations, including bitwise Boolean operations and several arithmetic operations, naturally satisfy this triangular dependency property when expanded at the bit level \cite{KlimovShamir2004}. This observation motivated a line of work on T-function-based cryptographic primitives and on their structural analysis \cite{KlimovShamir2005,Daum2005}.

A related, but distinct, form of topological deviation arises in data-dependent word operations. In RC5, for example, the amount of rotation is itself determined by intermediate data, so the effective local dependencies vary with the internal state \cite{Rivest1995}. Similar design choices also appear in MARS, whose use of data-dependent rotations contributes to a heterogeneous and non-standard internal organization \cite{NISTAESReport2001}. When such operations are expanded into Boolean circuits, they can be viewed as inducing value-controlled bit dependencies rather than a single fixed wiring pattern. The resulting circuit remains deterministic, but its effective propagation structure depends on internal control variables generated during the computation.

\paragraph{Cross-coupled dependencies and non-modular mixing.}

Several families of ciphers depart from strict blockwise modularity by introducing early dependencies between distant subblocks. In canonical SPNs, modularity is strong: nonlinear operations are localized in parallel S-boxes, while diffusion is delegated to a separate global linear layer. By contrast, some non-standard constructions create and recombine dependencies within the same phase, so that the boundary between local nonlinear processing and global mixing becomes less explicit.

This phenomenon is visible in several established design families. In generalized Feistel networks, the state is distributed over multiple branches whose update and permutation pattern induce richer inter-branch couplings than in the classical two-branch Feistel setting \cite{Nyberg1996,SuzakiMinematsu2010}. In ARX ciphers, modular additions, XORs, and rotations couple words through carries and word-level interactions, yielding dependency patterns that are not naturally decomposed into independent boxes followed by a separate diffusion layer \cite{DinuEtAl2016,BiryukovEtAl2016}. In LS-designs, the bitslice organization deliberately intertwines substitution and diffusion patterns in order to improve software masking efficiency and implementation regularity; subsequent work refined this approach and introduced extended LS variants with improved security-efficiency trade-offs \cite{GrossoEtAl2014,JournaultVariciStandaert2017}. At a finer implementation level, logic-minimized Boolean realizations of cryptographic components also show that compact subcircuits need not preserve a textbook gate-by-gate decomposition into clearly separated logical roles \cite{BoyarPeralta2010,BoyarPeralta2012}.

These examples show that non-standard topology may be motivated not only by cryptanalytic design choices, but also by implementation constraints such as compactness, regularity, and side-channel-aware efficiency.

\paragraph{Compressed circuits and distributed nonlinearity.}

A distinct but related theme is the progressive \emph{compression} of cryptographic functionality into compact Boolean subcircuits. This appears in hand-optimized realizations of core components such as S-boxes, both in compact hardware-oriented designs and in bitslice implementations \cite{Canright2005,RebeiroEtAl2006}. The Boyar--Peralta line of work then showed that standard components such as the AES S-box admit very compact gate-level realizations with low depth and low gate count, making explicit the extent to which a high-level cryptographic block can be rewritten as a tightly optimized Boolean network \cite{BoyarMatthewsPeralta2013,BoyarPeralta2012}.

This optimization perspective is directly relevant to emerging circuit topologies. Once a cryptographic component is expressed as a low-level Boolean network, the original conceptual boundaries between inversion, affine correction, linear post-processing, and auxiliary wiring may largely disappear. What remains is a compact directed graph of Boolean dependencies whose topology no longer mirrors the clean functional decomposition visible at the design level.

From this viewpoint, \emph{compressed nonlinearity} is not a separate primitive family, but rather a change of representation: a standard design may appear topologically non-standard once projected to the gate level. This observation matters for automated design, because a search procedure operating directly in circuit space may rediscover structures functionally close to known cryptographic components, but expressed as compact entangled subgraphs rather than as recognizable boxes and layers.

\paragraph{Feedback structures and recurrent cryptographic circuits.}

The strongest departure from the classical combinational block-cipher model is the introduction of feedback. In standard block-cipher descriptions, one evaluation of the primitive is represented by an acyclic combinational dependency graph. By contrast, many stream ciphers are inherently sequential: the next state is obtained from the current one through an iterative update of the form
\[
x^{(t+1)} = F(x^{(t)}).
\]
Such objects are therefore more naturally viewed as synchronous Boolean dynamical systems than as single combinational circuits.

Linear and nonlinear feedback shift registers provide the classical examples of this design principle. Trivium and the Grain family are emblematic in this respect: both rely on recurrent state updates and compact feedback functions rather than on a purely acyclic SPN- or Feistel-like round circuit \cite{DeCannierePreneel2008,HellEtAl2008}. More recent lightweight designs also illustrate the same shift in viewpoint. TinyJAMBU, for instance, is built around keyed permutations whose round function has an NLFSR-like structure with a very small nonlinear core, and whose security is studied at the level of the underlying iterative state transition rather than as a standard combinational round function \cite{SibleyrasEtAl2022}.

For the present discussion, the key point is that recurrent Boolean structures are already central in symmetric cryptography. Allowing non-acyclic Boolean graphs is therefore not an artificial generalization, but a natural way of aligning the theory of cryptographic Boolean circuits with an established family of feedback-based designs.

\paragraph{Cellular automata as local recurrent topologies.}

Cellular automata (CA) provide a particularly clear instance of non-standard topology based on locality and recurrence. A CA evolves by applying the same local update rule synchronously over an array of cells, so that global behavior emerges from uniform local interactions iterated over time. In cryptography, this model was explored early for pseudorandom generation, and its limitations were also analyzed from the cryptanalytic viewpoint, which already positioned CA primarily as recurrent generators rather than as one-shot combinational blocks \cite{Wolfram1986,MeierStaffelbach1991}.

More recent work has also investigated CA as a source of cryptographic components, especially through the study of bipermutive rules and CA-based S-boxes \cite{LeporatiMariot2014,MariotEtAl2019,GhoshalEtAl2018}. From a circuit perspective, the interest of CA lies in the combination of strict locality, structural regularity, and iterative state evolution. At the same time, this same locality implies that long-range dependencies can only emerge through repeated updates, which makes CA especially natural for pseudorandom generators and compact nonlinear components, rather than as direct substitutes for the wide diffusion layers of modern SPN block ciphers.

For the study of non-standard cryptographic Boolean circuits, CA therefore remain highly relevant: they provide a clean mathematical model of synchronous update, local constraints, and emergent global behavior.

\paragraph{Conditional routing and control-dependent wiring.}

A less standard, but conceptually important, departure from canonical round architectures arises when the effective propagation pattern depends on control variables. At the specification level, such mechanisms usually appear as data-dependent word operations rather than as explicit switch networks. A classical example is RC5, in which the amount of rotation is determined by intermediate data values \cite{Rivest1995}. When expanded at the Boolean level, this can be interpreted as a form of control-dependent local routing: the resulting dependency pattern remains deterministic, but it is no longer described by a single fixed wiring scheme.

Related phenomena also appear in more heterogeneous designs such as MARS. In the AES evaluation literature, MARS is repeatedly described as having a heterogeneous round structure and as relying in part on data-dependent rotations \cite{NISTAESReport2001}. From a circuit viewpoint, such mechanisms are naturally read as instances of conditional propagation, where part of the effective bit-level dependency graph is selected during the computation rather than fixed once and for all at design time.

This design direction remains much less standardized than SPN, Feistel, or ARX constructions. It is therefore best viewed not as a mature primitive family, but as a partially explored topological regime at the boundary between fixed wiring and value-controlled propagation.

\paragraph{Automated search, heuristic synthesis, and emergent structures.}

A recent methodological shift in symmetric cryptography is the growing use of automated tools for the analysis and local optimization of primitives and components. Mixed-integer linear programming, constraint programming, and SAT-based methods are now routinely used to search differential or linear characteristics, evaluate structural properties, and automate parts of the cryptanalytic workflow within fixed cipher families \cite{MouhaEtAl2012,SunEtAl2017,SunWW21,BelliniEtAl2024}. Logic minimization techniques have likewise been used to compress concrete Boolean components, showing that standard cryptographic blocks can admit compact low-level realizations with reduced gate count and depth \cite{BoyarMatthewsPeralta2013,BoyarPeralta2012}.

Most of this literature, however, still operates inside parameterized design spaces rather than over the full space of unrestricted Boolean circuits. In parallel, metaheuristic and evolutionary methods have produced a substantial body of work on the construction of Boolean functions satisfying cryptographic criteria, but usually at the function level rather than at the level of full circuit topology \cite{DjurasevicEtAl2023,PicekEtAl2015}. 

These developments remain directly relevant to non-standard topologies. Once optimization is performed over low-level Boolean descriptions rather than over fixed high-level templates, one naturally encounters irregular subgraphs, merged logical roles, compressed nonlinear structures, and asymmetric dependency patterns. In that sense, emergent topology is not yet a mature design family in its own right, but it already appears as a natural by-product of automated search in enlarged cryptographic design spaces.

\subsection{Systematic Exploration of Cryptographic Topologies}

The survey above suggests three main conclusions.

First, non-standard topology is not foreign to cryptography. It already appears in established families such as GFNs, Lai-Massey schemes, ARX ciphers, T-function-based constructions, NLFSR/LFSR-based stream ciphers, and cellular-automata-based designs.

Second, these families relax different structural constraints:
\begin{itemize}
	\item explicit separation of confusion and diffusion,
	\item homogeneous round repetition,
	\item modular decomposition into independent boxes,
	\item acyclicity of the evaluation graph,
	\item or static, state-independent wiring.
\end{itemize}
Hence, ``non-standard topology'' should not be understood as one single alternative model, but as a spectrum of departures from the canonical SPN/Feistel viewpoint.

Third, the literature remains fragmented. These families are usually studied in isolation and typically analyzed with their own vocabulary and proof techniques. As a consequence, it remains difficult to compare their structural properties and cryptographic robustness. What is still missing is a unifying circuit-level framework capable of describing these objects as instances of a broader class of cryptographic BCs, including architectures with interleaving, cross-coupling, conditional routing, depth heterogeneity, and possibly synchronous feedback. 

This gap motivates the present perspective: rather than treating such structures as isolated exceptions, one may consider them as points in a larger design space of cryptographic BCs, within which automated exploration can reveal previously unseen but nevertheless meaningful topologies.

Such a systematic and holistic analysis can initially be undertaken using analytical and mathematical methods. However, this approach is extremely challenging, and it is unclear whether cryptography possesses all the mathematical tools necessary to tackle this vast undertaking. The approach proposed in this paper is, rather, an empirical research protocol: we developed a genetic algorithm designed to optimize BCs against generic attacks, systematically for each of the 64 predefined architectural classes.

\paragraph{Exploration of under-studied circuit structures.}

A first objective of this framework is exploratory. By enumerating and analyzing families of circuits that are not tied to classical templates, the method makes it possible to investigate \emph{previously unexplored or weakly studied cryptographic topologies}. Such structures include circuits with interleaved linear and nonlinear components, cross-coupled dependency patterns, heterogeneous logical depths, or partially recurrent update rules. Because these circuits are generated under explicit structural constraints rather than manually designed, the exploration is not biased toward familiar architectures. This allows the discovery of circuit organizations that may not naturally arise from traditional design intuition.

\paragraph{Topology-level comparison of cryptographic structures.}

A second objective is comparative. Once circuits are grouped according to their structural constraints, the resulting classes of topologies can be analyzed collectively. This makes it possible to compare, at a structural level, the cryptographic behavior of different architectural families. For example, one may study how properties such as stratification, locality constraints, cross-block dependencies, or feedback influence diffusion, nonlinearity propagation, or algebraic complexity. Instead of analyzing individual cipher instances, the analysis operates at the level of \emph{topological classes of circuits}. This perspective complements classical cryptanalysis, which usually focuses on a single manually designed primitive.

\paragraph{Quantitative robustness evaluation via systematic topology exploration.}

The third objective is quantitative. To assess the cryptographic quality of each topological class, we use fitness functions that measure the robustness of the generated circuits with respect to several fundamental cryptanalytic dimensions: differential, linear, and algebraic resistances. The measurements obtained provide statistical estimates of cryptographic robustness for each of the 64 architectures studied, i.e., classes of cryptographic BCs. This allows for a quantitative comparison of topological design choices, instead of relying solely on qualitative arguments or isolated examples.

\paragraph{Toward a structural theory of cryptographic circuits.}

Taken together, these elements motivate a shift in perspective. Instead of treating block ciphers, stream ciphers, and other primitives as isolated algorithmic designs, the present approach views them as points in a broader landscape of cryptographic BCs. Systematic exploration of this landscape provides three complementary benefits:

\begin{itemize}
	\item identification of previously unexplored circuit topologies with promising cryptographic properties;
	\item structural comparison of architectural families that are usually studied independently;
	\item quantitative evaluation of the relationship between circuit topology and resistance to differential, linear, and algebraic cryptanalysis.
\end{itemize}

This topology-oriented perspective aims to complement traditional cipher design methodology by providing a more general framework for understanding how structural properties of BCs influence cryptographic security.

\section{Preliminaries and Definitions}  

\subsection{Booleans Circuits and Synchronous Boolean Networks}

\subsubsection{Logic Gates, Boolean Variables and Boolean Functions}

Let $\mathbb{B} = \{0,1\}$ denote the Boolean domain. A \emph{Boolean variable} is a variable taking values in $\mathbb{B}$, and a \emph{Boolean state} of dimension $n$ is a vector
\[
x = (x_1,\dots,x_n) \in \mathbb{B}^n.
\]
A \emph{logic gate} of arity $k$ is a map
\[
g : \mathbb{B}^k \to \mathbb{B}.
\]

The standard gates include negation $\neg$, conjunction $\wedge$, disjunction $\vee$, and exclusive-or $\oplus$. Any finite composition of such gates defines a Boolean transformation, and complete gate bases such as $\{\neg,\wedge,\vee\}$ or $\{\neg,\wedge,\oplus\}$ are sufficient to express arbitrary Boolean computations \cite{carletBooleanFunctionsCryptography2020}.

A \emph{scalar Boolean function} on $n$ variables is a map
\[
f : \mathbb{B}^n \to \mathbb{B},
\]
whereas a \emph{vectorial Boolean function} is a map
\[
F : \mathbb{B}^n \to \mathbb{B}^m,
\qquad
F(x) = \bigl(f_1(x),\dots,f_m(x)\bigr),
\]
with each coordinate function $f_j : \mathbb{B}^n \to \mathbb{B}$. In cryptography, substitution boxes, round functions, and more generally Boolean circuit realizations are naturally modeled as vectorial Boolean functions.

Boolean functions admit several equivalent representations, depending on whether one emphasizes exhaustive specification, structural realization, or algebraic analysis.

\paragraph{Truth tables.}
The most direct representation of a Boolean function $f : \mathbb{B}^n \to \mathbb{B}$ is its truth table, which lists the value of $f(x)$ for every input $x \in \mathbb{B}^n$. This representation is canonical and extensional: it completely determines the function, but its size grows exponentially with $n$, namely $2^n$ entries. For a vectorial function $F : \mathbb{B}^n \to \mathbb{B}^m$, one may equivalently use a table with $2^n$ rows and $m$ output columns.

\paragraph{Boolean graphs.}
A Boolean function may also be represented structurally by a labeled directed graph. In such a graph, source nodes represent input variables or constants, internal nodes represent logic gates, and designated sink nodes represent outputs. Edges encode functional dependence between intermediate quantities. When the graph is acyclic, this representation yields a classical combinational \emph{Boolean circuit}. When feedback edges are allowed, the same graph-theoretic language leads to Boolean networks, whose semantics must then be defined dynamically rather than purely combinationally. This viewpoint is central when studying topology, dependency structure, depth, fan-in, fan-out, or feedback patterns.

\paragraph{Algebraic Normal Form (ANF).}
Every scalar Boolean function $f : \mathbb{B}^n \to \mathbb{B}$ admits a unique polynomial representation over $\mathbb{F}_2$, called its \emph{Algebraic Normal Form}:
\[
f(x)
=
\bigoplus_{u \in \mathbb{B}^n}
a_u \prod_{j=1}^n x_j^{u_j},
\qquad a_u \in \mathbb{F}_2,
\]
with the convention $x_j^0 = 1$ and $x_j^1 = x_j$. This representation is unique because Boolean polynomials are considered modulo the relations $x_i^2 = x_i$. For a vectorial Boolean function
\[
F = (f_1,\dots,f_m) : \mathbb{B}^n \to \mathbb{B}^m,
\]
the ANF is defined coordinate-wise from the ANFs of the coordinate functions $f_j$. The ANF is especially important in cryptography, since it directly exposes monomial structure, algebraic degree, and multivariate dependencies.

\subsubsection{Acyclic Boolean Circuits}

A \emph{Boolean circuit} is a finite directed graph
\[
C = (V,E,I,O,\lambda),
\]
where $V$ is the set of nodes, $E \subseteq V \times V$ the set of directed wires, $I \subseteq V$ the set of input nodes, $O \subseteq V$ the set of designated output nodes, and $\lambda$ assigns to each non-input node $v \in V \setminus I$ a Boolean gate
\[
\lambda(v) : \mathbb{B}^{\deg^-(v)} \to \mathbb{B}
\]
compatible with its in-degree $\deg^-(v)$. The circuit is said to be \emph{acyclic} when its underlying directed graph contains no directed cycle.

For an input assignment $u \in \mathbb{B}^{|I|}$, a valuation
\[
\nu_u : V \to \mathbb{B}
\]
is said to be \emph{locally correct} if (i) $\nu_u$ coincides with $u$ on the input nodes; (ii) for every non-input node $v \in V \setminus I$ with predecessors $$\operatorname{pred}(v) = (p_1,\dots,p_k),$$ one has $$\nu_u(v)=\lambda(v)\bigl(\nu_u(p_1),\dots,\nu_u(p_k)\bigr).$$

In the classical combinational setting, global correctness is ensured by acyclicity. Indeed, an acyclic directed graph admits a topological ordering, so the values of all nodes can be computed by a finite and unambiguous forward evaluation. Hence, for every input assignment $u$, there exists a unique locally correct valuation $\nu_u$, and the circuit induces a well-defined Boolean map
\[
F_C : \mathbb{B}^{|I|} \to \mathbb{B}^{|O|},
\qquad
F_C(u) = \nu_u|_O.
\]

In this sense, the classical correctness criterion is fundamentally \emph{structural}: local gate consistency together with global acyclicity guarantees existence and uniqueness of the computed output. This is precisely the point that will later motivate the transition from acyclic BC to Boolean networks equipped with an explicit dynamical semantics.

\subsubsection{Boolean Networks and Synchronous Semantics}

A \emph{Boolean network} of dimension $n$ is a family of local update functions
\[
\Phi = (\phi_1,\dots,\phi_n),
\qquad
\phi_i : \mathbb{B}^n \to \mathbb{B},
\]
which together define a global transition map
\[
\Phi : \mathbb{B}^n \to \mathbb{B}^n,
\qquad
\Phi(x) = \bigl(\phi_1(x),\dots,\phi_n(x)\bigr).
\]
Equivalently, a Boolean network may be represented by a directed interaction graph whose nodes carry local Boolean update rules \cite{kauffman1969metabolic,richard2019fds}.

Under \emph{synchronous semantics}, all coordinates are updated simultaneously at discrete time steps:
\begin{equation}
	x(t+1) = \Phi\bigl(x(t)\bigr).
	\label{eq:sbn}
\end{equation}
A Boolean network equipped with this synchronous update rule is called a \emph{Synchronous Boolean Network} (SBN).

Thus, an SBN defines a deterministic discrete-time dynamical system over the finite state space $\mathbb{B}^n$. Since $\mathbb{B}^n$ is finite and $\Phi$ is a total map, every initial state $x(0)$ generates a unique trajectory
\[
x(0),x(1),x(2),\dots
\]
which eventually becomes periodic, possibly reaching a fixed point \cite{kauffman1969metabolic,aracena2009robustness}.

The key point is that the interaction graph of an SBN may contain directed cycles. Under synchronous semantics, such cycles are not pathological: they simply indicate that computation is no longer combinational and static, but dynamical and iterative.

From this perspective, the correctness criterion is no longer structural, as in classical acyclic circuits, but semantic and dynamical.

\paragraph{Local correctness.}
Each local rule
\[
\phi_i : \mathbb{B}^n \to \mathbb{B}
\]
must be a well-defined Boolean function. Equivalently, every node is assigned an unambiguous Boolean update rule compatible with its designated dependencies.

\paragraph{Global correctness.}
The global transition map
\[
\Phi : \mathbb{B}^n \to \mathbb{B}^n
\]
must be total and deterministic. Under this condition, every initial state determines a unique trajectory under the recurrence \eqref{eq:sbn}. Hence, the network is globally well-defined even when its interaction graph contains feedback cycles.

If, in addition, one specifies an initialization map
\[
\iota : \mathbb{B}^{n_{\mathrm{in}}} \to \mathbb{B}^n,
\]
a readout map
\[
\pi_O : \mathbb{B}^n \to \mathbb{B}^{n_{\mathrm{out}}},
\]
and an evaluation horizon $T \geq 1$, then the SBN induces a well-defined input--output transformation
\[
F_{\Phi,T}
=
\pi_O \circ \Phi^T \circ \iota.
\]
Thus, for SBNs, well-posedness means existence and uniqueness not of a static valuation obtained by topological propagation, but of a trajectory and of the output extracted from it after a prescribed number of synchronous iterations.

This contrasts with the classical combinational criterion for acyclic Boolean circuits, where global correctness follows from local gate consistency together with acyclicity of the wiring graph. In the SBN setting, acyclicity is no longer required: it is replaced by explicit dynamical semantics.

\subsubsection{Why Synchronous Boolean Networks in the Present Work}

The SBN formalism is adopted here because it strictly generalizes classical acyclic Boolean circuits while preserving the same local Boolean language. An acyclic circuit computes its output through a finite forward propagation along a topological order, whereas an SBN computes through repeated synchronous applications of a global transition map. The former may therefore be viewed as a special case of the latter, corresponding to a feedback-free dependency structure.

This generalization is important for the present work for three reasons.

First, it allows the exploration of \emph{non-acyclic topologies}. Restricting the design space to directed acyclic graphs excludes by construction all feedback-rich architectures. Such a restriction is natural in classical combinational logic, but it is unnecessarily narrow when the objective is to investigate cryptographic structures from a broader topological viewpoint.

Second, SBNs provide a \emph{well-posed semantics} in the presence of cycles. In a classical circuit, a directed cycle breaks the direct topological evaluation principle, so acyclicity is imposed as a sufficient condition for well-defined computation. In an SBN, by contrast, computation is defined by the synchronous recurrence \eqref{eq:sbn}, so no topological ordering is needed. The relevant global criterion is the existence of a uniquely defined evolution for every initial state.

Third, SBNs naturally match the logic of \emph{iterated cryptographic computation}. A synchronous update step can be interpreted as one round of computation, and the effective input--output map is obtained after a prescribed number $T$ of iterations. This framework is therefore particularly convenient for comparing standard and non-standard topologies under a common cryptographic evaluation protocol.

Therefore, acyclicity remains a sufficient correctness criterion in the classical circuit model, but it is no longer a necessary one in the SBN model. This distinction is central here: it allows us to study Boolean architectures whose computational meaning remains perfectly well-defined even though their topology is not acyclic. For this reason, SBNs provide the appropriate formal framework for the systematic exploration of unconventional cryptographic Boolean structures.

\subsection{Design Space: $2^6=64$ Architectures}

As we have seen, a BC can be constrained to be acyclic, which facilitates its analysis. But this is not the only constraint that can be imposed on BC. We now define the six constraints defining 64 architectures.

\subsubsection{Definition of the Six Constraints}


\paragraph{S -- Stratification.}
The network is said to be stratified in the SPN sense when there exist an integer \(d\ge 2\), a layer map $\ell:V\to\{1,\dots,d\}$, and a type map
$\tau:\{1,\dots,d\}\to\{\mathsf S,\mathsf P\}$, such that $\tau(k+1)\neq \tau(k) \qquad\forall k\in\{1,\dots,d-1\}$, and, for every vertex \(v\in V\),
$$v\in \ell^{-1}(k)\ \Longrightarrow\ f_v \text{ is of type } \tau(k).$$
Thus, stratification refers here to an alternating substitution--permutation layer structure, rather than to a feed-forward constraint on edges. In particular, this notion of stratification does not imply acyclicity. No constraint is imposed on the edge set with respect to the layer order.

\paragraph{A -- Acyclicity.}
The network is acyclic when its structural graph has no directed cycle:
\[
A(\mathcal{C})=1
\iff
G=(V,E)\text{ is a directed acyclic graph}.
\]
If \(A(\mathcal{C})=0\), the graph may contain feedback loops, and the SBN is then understood dynamically through parallel iteration of its global update map.

\paragraph{R -- Regularity.}
The network is regular when all input-to-output minimal propagation times coincide:
$$R(\mathcal{C})=1\iff
\exists T\in \mathbb{N}\ \text{such that}\ \forall x\in V_{\mathrm{in}},\ \forall y\in V_{\mathrm{out}},$$ $$\tau_G(x,y) := \min \{ t \ge 0 \mid y \text{ depends on } x \text{ after } t \text{ iterations}\}.$$

In the feed-forward case, when \(V_{\mathrm{in}}\) and \(V_{\mathrm{out}}\) are the source and sink sets, this reduces to the classical equal-depth condition.
In the recurrent case, it remains meaningful as a purely dynamical equal-delay condition.

\paragraph{I -- Interleaving.}
The network is interleaved when at least one edge crosses the prescribed block decomposition:
\[
I(\mathcal{C})=1
\iff
\exists (u,v)\in E\ \text{such that}\ b(u)\neq b(v).
\]
where $b:V \to \{1,\dots,B\}$ is a fixed block map. Equivalently,
\[
I(\mathcal{C})=0
\iff
\forall (u,v)\in E,\qquad b(u)=b(v).
\]
Thus \(I\) measures cross-block coupling independently of layering, acyclicity, or locality.

\paragraph{H -- Homogeneity.}
The network is homogeneous when each prescribed layer carries a unique local rule up to permutation of inputs:
\[
H(\mathcal{C})=1
\iff
\forall k\in \{1,\cdots,d\},\ \exists r_k\in \mathbb{N},\ \exists \varphi_k:\{0,1\}^{r_k}\to\{0,1\}
\]
such that every vertex \(v\in V_k\) has indegree \(r_k\) and
\[
\forall v\in V_k,\qquad f_v=\varphi_k\circ P_v
\]
for some permutation \(P_v\) of the \(r_k\) input coordinates.
Hence \(H\) is a symmetry condition on local update rules relative to the fixed layer partition, not a consequence of \(S\).

\paragraph{L -- Locality.}
The network is local when every interaction stays within the prescribed geometric radius:
$$L(\mathcal{C})=1 \iff \forall (u,v)\in E,\qquad \rho(u,v)\le \delta_0.$$
where $(V,\rho)$ is a fixed metric space. Because \(\delta_0\) is fixed in advance, \(L\) is a genuine design predicate and not a vacuous bounded-diameter statement.

\medskip

We prove below that these six predicates are logically independant. In particular, no predicate is included within another, and any observed correlation between them is empirical rather than logical. In doing so, the predicate set $\Omega := \{S, A, R, I, H, L\}$ forms a logically independent set of predicates.

\subsubsection{Motivation}

The predicate set $\Omega := \{S, A, R, I, H, L\}$ should be understood as a structural coordinate system on the space of well-formed SBN architectures. These predicates are not security criteria in themselves. Rather, they define architectural priors whose cryptographic effects can subsequently be evaluated through differential, linear, or algebraic indicators. This separation between \emph{construction variables} and \emph{security observables} is methodologically important: it prevents the search space from being hard-coded around a single attack model, while keeping the description of architectures sufficiently abstract to cover both classical and non-standard designs.

A useful exploratory basis must satisfy four requirements. First, it must be \emph{general}: the constraints should apply to layered feed-forward circuits as well as to recurrent SBNs. Second, it must be \emph{interpretable}: each constraint should correspond to a distinct structural mechanism. Third, it should be \emph{minimally redundant}: the descriptive power should not come from hidden reformulations of the same property. Fourth, it must be \emph{composable}: admissible conjunctions of constraints should define meaningful subclasses of SBNs. The present family is designed to satisfy these four requirements.

\paragraph{(i) Stratification.}
$S$, i.e.\ S/P typing, isolates architectures organized by an alternating substitution--permutation pattern. It captures the canonical SPN design grammar without imposing any constraint on edge orientation or propagation. Its role is not to enforce feed-forward depth, but to encode the alternation between nonlinear confusion layers and linear diffusion layers. Thus, \(S\) defines a natural reference class for comparing more general Boolean-network architectures, including non-standard and recurrent topologies.

\paragraph{(ii) Acyclicity.}
Acyclicity is the primary bifurcation between combinational and dynamical behavior. When \(A=1\), the network computes by finite propagation on a DAG. When \(A=0\), the same object becomes a genuine synchronous dynamical system governed by parallel iteration \(x(t+1)=F(x(t))\). Keeping \(A\) explicit is therefore essential if one wants to treat feedback not as a pathology, but as a first-class architectural degree of freedom.

\paragraph{(iii) Regularity.}
Regularity controls the isotropy of information propagation. In the acyclic setting, equal input-to-output minimal propagation lengths eliminate structural latency bias and prevent circuit depth from being confounded with routing irregularities. In the recurrent setting, the dynamic reinterpretation of \(R\) as equal minimal propagation time between designated boundary sets preserves the same design intent: all parts of the network should communicate on comparable timescales. Thus \(R\) is not a constraint on local logic, but on the temporal symmetry of influence propagation.

\paragraph{(iv) Interleaving.}
Interleaving acts at the mesoscopic scale defined by a block decomposition. It distinguishes architectures in which dependencies remain confined inside blocks from architectures in which information crosses block boundaries. This is a structurally important axis because modularity and mixing are not opposites of depth or regularity: a network may be acyclic, stratified, and perfectly regular, while still being either strongly compartmentalized or strongly cross-coupled. In cryptographic terms, \(I\) influences the extent to which nonlinear effects can escape local compartments and spread across the architecture.

\paragraph{(v) Homogeneity.}
Homogeneity constrains the repertoire of local Boolean functions rather than the graph itself. It captures the difference between replicated architectures, in which one reuses the same local rule up to permutation of inputs, and heterogeneous architectures, in which distinct rules coexist. This distinction is essential for a general theory because symmetry can be either an asset or a liability: it may simplify design, analysis, and implementation, but it may also introduce structural regularities exploitable by an attacker. Since \(H\) acts on node labels rather than on edges, it operates at a different descriptive level from the topological constraints.

\paragraph{(vi) Locality.}
Locality introduces an ambient metric and constrains the span of interactions. This is indispensable if the model is meant to cover both short-range architectures, where diffusion emerges from repeated local propagation, and long-range architectures, where a few non-local connections induce rapid global mixing. The role of \(L\) is therefore twofold: it encodes implementation realism (wiring span, layout constraints, communication cost) and it distinguishes qualitatively different mechanisms of diffusion.

\paragraph{Reversibility.}
At the level of practical cipher instantiation, an additional issue arises: if one wishes to build a permutation on a fixed-size state space, a network composed of intrinsically non-injective gates such as \texttt{AND} and \texttt{OR} cannot be globally bijective unless it is embedded in a compensating reversible construction. In the present setting, a balanced Feistel structure provides a practical compromise, since it restores bijectivity at the 16-bit level while keeping 8-bit nonlinear components computationally tractable.

Under this compromise, the six structural constraints remain qualitatively meaningful, but not all of them retain the same granularity at the global 16-bit level. In particular, the interleaving predicate \(I\) becomes partly absorbed by the mandatory Feistel exchange between the two 8-bit branches: cross-branch mixing is always present at that scale. By contrast, \(I\) remains fully discriminative inside each 8-bit branch, which is its natural mesoscopic level in this implementation.

A larger 32-bit setting would have preserved the natural 16-bit granularity of all six constraints, but at prohibitive evaluation cost: exact computation of the full truth table, Walsh--Hadamard transform, difference distribution table, and algebraic normal form would require exhaustive processing of \(2^{32}\) inputs. We deliberately avoided Monte Carlo estimation in order not to inject statistical noise into the fitness signal used by the genetic algorithm. Another possible alternative would have been to replace non-reversible gates by reversible primitives such as 3-bit Toffoli gates, but this would substantially alter the semantic content of the \(S\) and \(H\) predicates by changing the nature of the local rule repertoire.

\paragraph{Composability.}
For any admissible subset $\Sigma \subseteq \Omega$, one may define the corresponding structural class by
\[
\mathcal{C}_{\Sigma} = \bigcap_{X\in\Sigma} \mathcal{C}_X,
\]
where \(\mathcal{C}_X\) denotes the class of well-formed SBNs satisfying constraint \(X\). This construction is methodologically valuable: it allows one to add or remove one structural prior at a time, compare adjacent architectural classes, and quantify interaction effects between constraints without changing the semantic definition of an SBN.

Thus the exploration space is not an arbitrary taxonomy but a controlled family of structurally meaningful classes. It is small enough to support systematic comparison, yet rich enough to include classical layered SPN-like architectures, irregular feed-forward topologies, and genuinely recurrent SBNs within the same formal language.

\paragraph{Cryptographic relevance.}
From the viewpoint of cryptographic research, this constraint basis has the right granularity. It is sufficiently coarse not to prejudge which architectural family is ``best'', and sufficiently fine to isolate the main structural levers that may influence resistance to differential, linear, and algebraic attacks: feedback versus feed-forward causality, balanced versus anisotropic propagation, modular separation versus cross-coupling, functional symmetry versus heterogeneity, and local versus non-local diffusion. $\Omega$ therefore provides a principled and general coordinate system for the exploration of SBNs, including non-standard topologies that would remain invisible in a purely DAG-based framework. The constraints act on different structural strata:
\begin{itemize}
	\item[(a)] causal structure and propagation geometry \((A,R)\),
	\item[(b)] mesoscopic coupling between blocks \((I)\),
	\item[(c)] functional layer organization and local rule symmetry \((S,H)\),
	\item[(d)] spatial or geometric interaction span \((L)\).
\end{itemize}

The architectural constraints thus cover sufficiently diverse conceptual aspects to support a broad comparative exploration.

\subsubsection{Logical Independence}

The family $\Omega$ is \emph{logically independent}. 

\begin{proposition}[Logical independence]
	Let $\Omega = \{S,A,R,I,H,L\}$ be the six predicates defined on enriched SBN. Then they are logically orthogonal in the following sense: for every $X\in\Omega$, there exist two enriched SBNs \(\mathcal C_X^+\) and \(\mathcal C_X^-\) such that $X(\mathcal C_X^+)=1$, $X(\mathcal C_X^-)=0$, while $Y(\mathcal C_X^+)=Y(\mathcal C_X^-)$ for all $Y\neq X$.	Consequently, no predicate among \(\Omega\) is logically implied by the conjunction of the other five.
\end{proposition}

\begin{proof}\label{proof}
	See appendix.
\end{proof}

\begin{corollary}
	Since the six predicates are logically independent,
	their joint valuation spans the Boolean cube $\{0,1\}^6$.
	Hence the architectural design space contains $2^6 = 64$ distinct classes.
\end{corollary}

Under the formal model above, the structural constraint set $\Omega$ defines six logically independent axes in the circuit design space. Their joint valuation spans the full Boolean cube $\{0,1\}^6$, yielding $2^6=64$ structural classes, which we call architectures.

This logical independence justifies treating the constraint system as a factorial design over the architectural fitness landscape.

\paragraph{Logical Independence, Empirical Dependence, and Local Effects}

The lattice-based validation does not indicate that the design constraints are biased, nor that they fail to be logically independent. Rather, it reveals a distinction between three fundamentally different levels of analysis.

\subparagraph{1. Logical (Combinatorial) Independence.}

Since every valuation in \(\{0,1\}^6\) is realizable, the design space forms a complete Boolean hypercube of size \(64\). No structural restriction forbids any constraint from co-occurring with any other. Therefore, the six predicates are \emph{logically independent}: the design space is combinatorially exhaustive.

\subparagraph{2. Statistical Dependence Within High-Performance Subsets.}

An FCA implication such as $\{X\}\Rightarrow \{Y\}$, $X,Y\in\Omega$, does not express a logical implication in the full design space. It only states that, within a selected high-performance subset, architectures satisfying \(X=1\) also tend to satisfy \(Y=1\). Formally, this means that
$\Pr(Y=1\mid X=1,\ \mathrm{elite})$ vs $\Pr(Y=1\mid \mathrm{elite})$, or, more generally, that \(Y\) is over-represented among elite architectures satisfying \(X\). This is a conditional statistical regularity inside a restricted region of the hypercube. Such associations may emerge even when the underlying design variables are fully independent at the combinatorial level. Hence, the implication reflects a \emph{context-dependent association}, not a structural constraint of the model.

\subparagraph{3. Local Effects on the Hypercube.}

Statistical dependence answers the question:
\begin{quote}
	Among high-performing architectures, do those with \(X=1\) also tend to have \(Y=1\)?
\end{quote}
But it does not answer:
\begin{quote}
	What is the effect of switching \(Y\) while keeping the rest of the architecture fixed?
\end{quote}

This second question is local and intervention-like. On an edge of the hypercube, it is measured by a quantity of the form $\Delta_Y(\mathcal C)=f(\mathcal C^{(Y=1)})-f(\mathcal C^{(Y=0)})$,
where the two architectures differ only in the value of constraint \(Y\). Thus, one moves from a global conditional association inside the elite subset to a local structural effect measured along an edge of the cube.

Consequently, the constraints remain combinatorially independent, while their fitness effects may be statistically dependent and strongly interaction-dependent.

Logical independence at the structural level therefore coexists with empirical dependence at the performance level. This distinction is central: the observed FCA rules reveal emergent interaction patterns in the fitness landscape, not bias or logical redundancy among the architectural constraints.



\subsubsection{Intuition and visualization}

Figure \ref{intuition} depicts the intuitive idea behind the structural constraints.

\begin{figure}[H]
	\centering
	\begin{subfigure}[t]{0.48\textwidth}
		\centering
		\includegraphics[page=1,width=55mm]{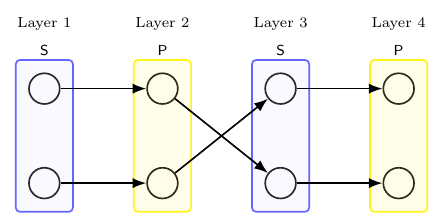}
		\caption*{[S] \textbf{Stratification}. Layers alternate between nonlinear substitution type (\(\mathsf{S}\)) and linear/affine permutation--diffusion type (\(\mathsf{P}\)).}
	\end{subfigure}
	\hfill
	\begin{subfigure}[t]{0.48\textwidth}
		\centering
		\includegraphics[page=2,width=55mm]{img/figures.pdf}
		\caption*{[A] \textbf{Acyclicity}. The graph admits a topological order: arrows always move forward in the causal ordering, and no structural feedback loop is present.}
	\end{subfigure}
	
	\vspace{1.5em}
	
	\begin{subfigure}[t]{0.48\textwidth}
		\centering
		\includegraphics[page=3,width=50mm]{img/figures.pdf}
		\caption*{[R] \textbf{Regularity}. All input-output minimal propagation times are equal. Regularity removes latency bias between distinct routes through the network.}
	\end{subfigure}
	\hfill
	\begin{subfigure}[t]{0.48\textwidth}
		\centering
		\includegraphics[page=4,width=55mm]{img/figures.pdf}
		\caption*{[I] \textbf{Interleaving}. At least one dependency crosses the prescribed block decomposition. The highlighted edge couples two otherwise distinct modules.}
	\end{subfigure}
	
	\vspace{1.5em}
	
	\begin{subfigure}[t]{0.48\textwidth}
		\centering
		\includegraphics[page=5,width=45mm]{img/figures.pdf}
		\caption*{[H] \textbf{Homogeneity}. Inside a given layer, all vertices implement the same Boolean rule (up to permutation of input wires).}
	\end{subfigure}
	\hfill
	\begin{subfigure}[t]{0.48\textwidth}
		\centering
		\includegraphics[page=6,width=55mm]{img/figures.pdf}
		\caption*{[L] \textbf{Locality}. Every edge must stay within the prescribed geometric radius $\delta_0$. The grey dashed arrow indicates a non-local shortcut excluded by the constraint.}
	\end{subfigure}
	\caption{Intuitive graph sketches for the six orthogonal predicates --- Each panel shows a canonical graph pattern illustrating the meaning of the predicate when it is satisfied.}
	\label{intuition} 
\end{figure}

\subsection{Reward and Fitness Functions}

In the proposed evolutionary framework, \emph{fitness functions} and \emph{reward functions} play two distinct but complementary roles. Fitness functions provide the global evaluation criteria used for selection, whereas reward functions supply local guidance signals that facilitate exploration of the search space.

\subsubsection{Fitness Functions}

In a cryptographic setting, the most meaningful fitness indicators are those directly derived from well-established cryptanalytic attacks. To remain applicable to a broad class of BCs---including S-boxes, nonlinear Boolean functions, and small round functions---we deliberately restrict the fitness family to three generic adversarial criteria:

\begin{itemize}
	\item[\textbf{F1}] Differential resistance
	\item[\textbf{F2}] Linear resistance
	\item[\textbf{F3}] Resistance to algebraic attacks
\end{itemize}

See Appendix for mathematical definitions. These three criteria correspond to the main paradigms of modern symmetric cryptanalysis and can be defined for arbitrary BCs without requiring specific assumptions on the internal organization of the circuit.

The selected fitness functions satisfy four key requirements for evolutionary cryptographic design:

\begin{itemize}
	\item \textbf{Attack-oriented evaluation.} Each fitness is directly associated with a well-identified cryptanalytic paradigm.
	\item \textbf{Model independence.} The definitions rely only on the Boolean input-output behavior of the circuit and do not presuppose any particular internal structure.
	\item \textbf{Applicability to generic BCs.} The metrics apply uniformly to S-boxes, Boolean networks, round functions, and other small cryptographic components.
	\item \textbf{Complementarity.} Differential and linear criteria capture statistical resistance, whereas algebraic resistance captures symbolic complexity.
\end{itemize}

Together, these three criteria form a minimal adversarial fitness core for evaluating the cryptographic quality of generic BCs. This core keeps the evolutionary selection focused on fundamental cryptanalytic hardness while preserving broad applicability across BC designs.

\subsubsection{Reward Functions}

Whereas fitness functions assess the global cryptographic quality of a candidate BC, they are often expensive to compute and may induce large plateaus in the search space. For this reason, the evolutionary process also relies on auxiliary \emph{reward functions}, whose role is not to measure security directly, but to provide smooth local signals indicating whether a mutation moves the circuit toward structurally harder regions of the design space.

Formally, let \(F : \mathbb{F}_2^n \rightarrow \mathbb{F}_2^m\) be a BC and let \(F'\) denote a mutated version of \(F\). Each reward is defined as the variation of a structural indicator between \(F\) and \(F'\). We restrict the reward design to six indicators capturing algebraic complexity, statistical uniformity, and structural entanglement:

\begin{itemize}
	\item[\textbf{R1}] Effective algebraic degree reward
	\item[\textbf{R2}] Algebraic normal form entropy
	\item[\textbf{R3}] Walsh spectrum flattening
	\item[\textbf{R4}] Local differential uniformization
	\item[\textbf{R5}] Dependency graph expansion
	\item[\textbf{R6}] Symmetry breaking reward
\end{itemize}

See Appendix for mathematical definitions. This reward core captures three complementary aspects of cryptographic hardness:

\begin{itemize}
	\item \textbf{Algebraic complexity.} Rewards R1 and R2 promote deep nonlinear interactions and dense algebraic representations.
	\item \textbf{Statistical uniformity.} Rewards R3 and R4 reduce the emergence of dominant linear and differential structures.
	\item \textbf{Structural entanglement.} Rewards R5 and R6 encourage strong variable mixing and penalize symmetric or repeated substructures.
\end{itemize}

Together, these indicators provide smooth local gradients while remaining distinct from the adversarial fitness functions that measure actual resistance to cryptanalysis.

\section{Methodology} 

We consider the complete constrained architecture space $\Omega = \mathbb{B}^n$, where each configuration $\mathbf{c}=(C_1,\dots,C_n)\in\Omega$ is assigned a scalar fitness value through a function $f:\mathbb{B}^n\to\mathbb{R}$. The methodology consists of six stages, from local effect analysis to global variance attribution.

\subsection{Stage 0: Multi-Level Discretization}

When categorical analyses are required, the scalar fitness is discretized into performance classes:

\begin{itemize}
	\item \textbf{Binary performance} (\texttt{Fit\_bin}): split at the median fitness value;
	\item \textbf{Elite performance} (\texttt{Fit\_elite}): restriction to the top $10$--$20\%$ of SBNs.
\end{itemize}

These two discretizations respectively capture coarse structural contrasts and highly selective architectural regimes.

\subsection{Stage 1: Hypercube Edge Analysis}

For each constraint $C_i$, we evaluate its marginal effect over all $2^{n-1}$ contexts induced by the remaining constraints. Writing $C_{-i}\in\mathbb{B}^{n-1}$, define
\[
\Delta f_i(C_{-i}) = f(C_{-i},C_i=1)-f(C_{-i},C_i=0).
\]

The resulting family $\{\Delta f_i(C_{-i})\}_{C_{-i}\in\{0,1\}^{n-1}}$ is summarized by:

\begin{itemize}
	\item the mean effect $\bar{\Delta}_i = 2^{-(n-1)}\sum_{C_{-i}}\Delta f_i(C_{-i})$;
	\item the effect variance $\sigma^2_{\Delta_i}=2^{-(n-1)}\sum_{C_{-i}}\left(\Delta f_i(C_{-i})-\bar{\Delta}_i\right)^2$;
	\item the sign stability, i.e., the proportion of contexts for which $\Delta f_i(C_{-i})>0$.
\end{itemize}

Large $\sigma^2_{\Delta_i}$ relative to $|\bar{\Delta}_i|$ indicates context dependence.

\subsection{Stage 2: Epistasis Matrix}

For each pair $(C_i,C_j)$ and each context $C_{-\{i,j\}}\in\mathbb{B}^{n-2}$, define
\[
\varepsilon_{ij}(C_{-\{i,j\}})=f_{11}-f_{10}-f_{01}+f_{00},
\]
where $f_{ab}=f(C_i=a,C_j=b,C_{-\{i,j\}})$.

The pairwise interaction is summarized by
\begin{align*}
	\bar{\varepsilon}_{ij} &= 2^{-(n-2)}\sum_{C_{-\{i,j\}}}\varepsilon_{ij}(C_{-\{i,j\}}),\\
	\sigma^2_{\varepsilon_{ij}} &= 2^{-(n-2)}\sum_{C_{-\{i,j\}}}\left(\varepsilon_{ij}(C_{-\{i,j\}})-\bar{\varepsilon}_{ij}\right)^2.
\end{align*}

This yields the epistasis matrix $\mathbf{E}=[\bar{\varepsilon}_{ij}]_{i,j=1}^n$.

Stable positive values identify synergistic pairs; stable negative values identify antagonistic pairs; large variances indicate context-dependent interactions.

\subsection{Stage 3: Formal Concept Analysis (FCA)}

To extract minimal structural patterns associated with high fitness, we define, for each discretization of Stage~0, a formal context $\mathbb{K}=(G,M,I)$, with
\begin{itemize}
	\item $G=\Omega=\mathbb{B}^n$;
	\item $M=\mathcal{C}\cup\{\text{high-performance}\}$, where $\mathcal{C}=\{C_1,\dots,C_n\}$;
	\item $I\subseteq G\times M$ the incidence relation.
\end{itemize}

A formal concept is a pair $(A,B)$ such that
\begin{align*}
	A' &= \{m\in M\mid \forall g\in A,\ (g,m)\in I\}=B,\\
	B' &= \{g\in G\mid \forall m\in B,\ (g,m)\in I\}=A.
\end{align*}

The concept lattice $\mathfrak{B}(\mathbb{K})$ is ordered by
\[
(A_1,B_1)\leq(A_2,B_2)\Longleftrightarrow A_1\subseteq A_2
\qquad
(\text{equivalently } B_2\subseteq B_1).
\]

From each lattice, we extract:

\begin{itemize}
	\item minimal generators: minimal subsets $B\subseteq\mathcal{C}$ such that $B\cup\{\text{high-performance}\}$ is closed;
	\item association rules $B_1\to B_2$ with support
	\[
	\mathrm{sup}(B_1\cup B_2)=\frac{|B_1'\cap B_2'|}{|G|}
	\]
	and confidence
	\[
	\mathrm{conf}(B_1\to B_2)=\frac{|B_1'\cap B_2'|}{|B_1'|}.
	\]
\end{itemize}

These objects provide minimal structural signatures of high-performance regions.

\subsection{Stage 4: Post-Lattice Validation}

Each pattern $\mathcal{M}\subseteq\mathcal{C}$ extracted at Stage~3 is tested on the restricted subcube
\[
(\mathbb{B}^n)^{\mathcal{M}}=\{C\in\mathbb{B}^n\mid \forall C_i\in\mathcal{M},\ C_i=1\}.
\]

For each $C_j\notin\mathcal{M}$, we evaluate whether adding $C_j$ significantly improves fitness by applying a Wilcoxon signed-rank test to the paired differences $\Delta f_j(C_{-j})$ within the restricted subspace. If no significant improvement is detected ($p>0.05$), then $C_j$ is treated as locally neutral relative to $\mathcal{M}$.

This step distinguishes structurally sufficient patterns from incomplete ones.

\subsection{Stage 5: Spectral Variance Decomposition on the Hypercube}

Since $f$ is observed on the full hypercube, all interaction orders are exactly identifiable. We therefore use an orthogonal Walsh--Hadamard decomposition.

\subsubsection{Stage 5a: Walsh--Hadamard Decomposition}

We re-encode each binary variable $C_i \in \{0,1\}$ into centered coordinates: $x_i = 2C_i - 1 \in \{-1,+1\}$. The fitness function admits a unique expansion over the Walsh basis:

\begin{equation*}
	f(x)=\sum_{S\subseteq\{1,\dots,n\}}\hat{f}(S)\chi_S(x), \qquad \chi_S(x)=\prod_{i\in S}x_i,
\end{equation*}
with coefficients $\hat{f}(S) = 2^{-6} \sum_{x \in \{-1,1\}^6} f(x)\,\chi_S(x)$.

Orthogonality gives $\sum_{x\in\{-1,1\}^n}\chi_S(x)\chi_T(x)=2^n\delta_{ST}$, hence
\[
\mathrm{Var}(f)=\sum_{S\neq\varnothing}\hat{f}(S)^2.
\]

For each order $k$, define the contribution of interaction order $k$ and the normalized epistatic spectrum, respectively
\[
V_k=\sum_{|S|=k}\hat{f}(S)^2,
\qquad
\eta_k^2=\frac{V_k}{\sum_{m=1}^n V_m}.
\]

Thus:
\begin{itemize}
	\item $k=1$ corresponds to additive effects;
	\item $k=2$ to pairwise epistasis;
	\item $k\geq 3$ to higher-order interactions.
\end{itemize}

\subsubsection{Stage 5b: Restricted Hierarchical Model}

For interpretability, we introduce a truncated descriptive model retaining:
\begin{enumerate}
	\item all first-order terms;
	\item the dominant second-order terms;
	\item indicator terms associated with selected FCA patterns $\mathcal{M}_k$.
\end{enumerate}

This yields
\begin{equation}
	\tilde{f}(x)=\mu
	+\sum_{|S|=1}\hat{f}(S)\chi_S(x)
	+\sum_{S\in\mathcal{S}_2}\hat{f}(S)\chi_S(x)
	+\sum_k \alpha_k\,\mathbb{I}[x\in\mathcal{M}_k],
\end{equation}
where $\mathcal{S}_2$ is the set of dominant pairwise terms.

This model is descriptive only. Exact variance attribution remains given by the Walsh spectrum.

\subsubsection{Final Artifact: Contribution Spectrum}

For each nonempty component $S$, define
\[
\eta_S^2=\frac{\hat{f}(S)^2}{\sum_{T\neq\varnothing}\hat{f}(T)^2}.
\]

This yields a complete contribution spectrum comprising:
\begin{itemize}
	\item additive effects ($|S|=1$),
	\item pairwise interactions ($|S|=2$),
	\item higher-order couplings ($|S|\geq 3$).
\end{itemize}

FCA-derived patterns are reported separately through their support, confidence, and conditional fitness gains.

\subsection{Methodological Synthesis}

The methodology proceeds from local to global structure:

\begin{itemize}
	\item Stage~1: marginal effects;
	\item Stage~2: pairwise epistasis;
	\item Stage~3: minimal structural rules;
	\item Stage~4: local validation of these rules;
	\item Stage~5: exact variance decomposition by interaction order.
\end{itemize}

It therefore connects exhaustive combinatorial exploration, lattice-based structural abstraction, and exact spectral attribution within a single framework. The methodology progresses from \emph{combinatorial exhaustivity} (Stages 0--2) through \emph{structural abstraction} (Stages 3--4) to \emph{quantitative attribution} (Stage 5).


\section{Results}   

\subsection{Differential Resistence}

This analysis for differential fitness produced converging evidence across marginal effects (Step 1), pairwise interactions (Step 2),  formal concept rules (Step 3), and full-dataset validation (Step 4). Six structural findings stand out  ---  and several directly invert the conclusions from linear resistance.

\paragraph{R is the universal mandatory driver.}
R (Regularity) is the only constraint with a stable positive individual effect
($\mu = +0.154$, $\mathrm{CV} = 0.77 < 1$, 30/31 hypercube edges positive). Its individual lift is $3.50$: when R is active, the Fit\_bin rate rises to $88\%$ versus $25\%$ without it. The FCA stem base yields the rule $\varnothing \Rightarrow \texttt{R}$ with support $= 1.00$ over the 9 elite architectures --- the strongest possible implication. All 15 top-ranked combinations contain R.

R creates uniform propagation depth, which is the structural prerequisite for
differential uniformity: input differences spread uniformly across all output bits only when every path has the same length. \emph{R should always be activated for this fitness objective.}

\paragraph{S+R is the dominant epistatic antagonism.}
The pair S+R has $\mathrm{SNR} = 1.71$ --- the only pair in the study with
$\mathrm{SNR} > 1$, making it a context-independent, fully reliable signal. The mean epistasis is $\varepsilon = -16.39\%$ across all 16 unit squares (16/16 antagonistic, 0/16 synergistic). The same antagonism exists in linear resistance ($\mathrm{SNR} = 0.95$) but is almost twice as strong here.

Stratification imposes asymmetric layer structure that directly conflicts with R's uniform depth discipline. \emph{S and R should never be combined for differential fitness.}

\paragraph{L is individually neutral, epistatically positive.}
L (Locality) has near-zero individual effect ($\mu = -0.011$, $\mathrm{CV} = 6.60$,
$\mathrm{Lift} = 1.00$ exactly). Yet the top three architectures are \texttt{A+R+L} $(1.000)$, \texttt{R+L} $(0.992)$, \texttt{A+R+H+L} $(0.982)$. L's contribution is purely epistatic: it amplifies R but contributes nothing alone.

For differential uniformity, bounded connection distance reinforces R's uniform differential propagation rather than obstructing it. \emph{Activate L in combination with R; never alone.}

\paragraph{H and I are individual degraders; H synergizes with S.}
Both H (Homogeneity) and I (Interleaving) show consistent negative individual effects
(\texttt{H}: $\mu = -0.061$, $\mathrm{CV} = 1.25$, 25/31 negative edges; \texttt{I}: $\mu = -0.052$, $\mathrm{CV} = 1.60$, 21/31 negative edges) and both have $\mathrm{Lift} = 0.80$.
However, the pair S+H shows $\mathrm{SNR} = 1.14$ ($\varepsilon = +9.52\%$, 14/16 synergistic) --- H reverses its role when combined with S.

\texttt{H} and \texttt{I} disrupt \texttt{R}'s uniform differential propagation. \texttt{S+H} forms a compensating S/P-box structure where \texttt{H}'s layer uniformity reinforces S's alternation rhythm. \emph{H and I should be avoided in the R pathway; H is viable in the S pathway only.}

\paragraph{A  is a pure contextual amplifier.}
\texttt{A} (Acyclicity) has $\mu \approx 0$ ($\mathrm{CV} = 30.28$) and $\mathrm{Lift} = 0.89$ --- no individual effect, slightly negative alone. Yet \texttt{A+R} is the top combination by mean score ($0.8419$) and elite rate ($0.38$). The FCA elite rule $\mathtt{I+R} \Rightarrow \mathtt{A}$ shows that within the \texttt{R} backbone, adding \texttt{I} implies A.

A eliminates feedback cycles that would break differential uniformity in an otherwise well-structured circuit. \emph{Use \texttt{A} as a refinement on top of \texttt{R}, never in isolation.}

\paragraph{Two design strategies.}
Two structurally independent paths reach high differential fitness: (1)~the \textbf{R-backbone}: R $+$ A $+$ L ($\pm$ H), best architecture \texttt{A+R+L} (score $1.000$), governed by the mandatory rule $\varnothing \Rightarrow \texttt{R}$; (2)~the \emph{S-pathway}: S $+$ H as a secondary micro-strategy (best \texttt{S+H}, score $0.820$). The R-backbone dominates completely: all 9 elite architectures contain R.

\paragraph{Summary tables.}

\begin{table}[H]
	\centering \rowcolors{2}{white}{gray!15} \small
	\begin{tabular}{cp{0.7cm}p{2.7cm}p{6cm}} \hline
		\textbf{Order} & \textbf{Var} & \textbf{Leading terms} & \textbf{Interpretation} \\ \hline
		1 & \textbf{64.1} & $R,\ S,\ H,\ I$ & Strongly additive; $R$ alone accounts for 44.3\%, $S$ and $H$ partially cancel \\
		2 & 22.1 & $S{\times}R$, $S{\times}H$, $R{\times}I$ & $S{\times}R$ antagonism dominates (12.7\%);	$S{\times}H$ synergy partially offsets \\
		3 & 7.6 & $R{\times}I{\times}L$, $A{\times}R{\times}I$, $A{\times}H{\times}L$
		& Diffuse ternary structure; $L$ enters the landscape only at order~3 \\
		4 & 5.2 & $S{\times}A{\times}I{\times}L$, $S{\times}R{\times}I{\times}L$
		& Residual four-way interactions involving $I$ and $L$ \\
		5 & 1.1 & $S{\times}R{\times}I{\times}H{\times}L,\ \ldots$ & Near-negligible high-order terms \\
		6 & 0.0 & $S{\times}A{\times}R{\times}I{\times}H{\times}L$ & Fully-constrained interaction; negligible \\
	\end{tabular}
	\caption{Spectral variance (\%) decomposition by interaction order for differential fitness}
\end{table}

\begin{figure}[H]
	\centering
	\includegraphics[width=\textwidth]{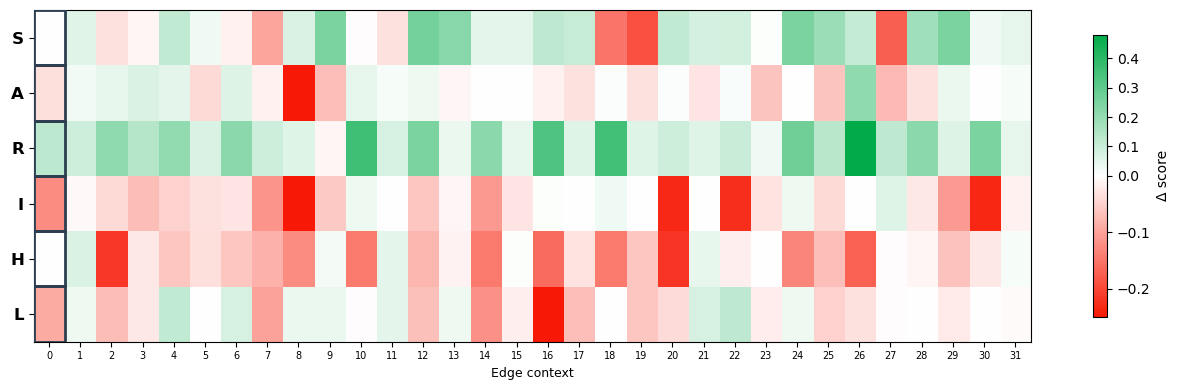}
	\caption{Sign of $\Delta C_i$ across 32 contexts (green = +, white = 0, red = --)}
	\label{heatmap_dif}
\end{figure}

\begin{figure}[H]
	\centering
	\includegraphics[width=6cm]{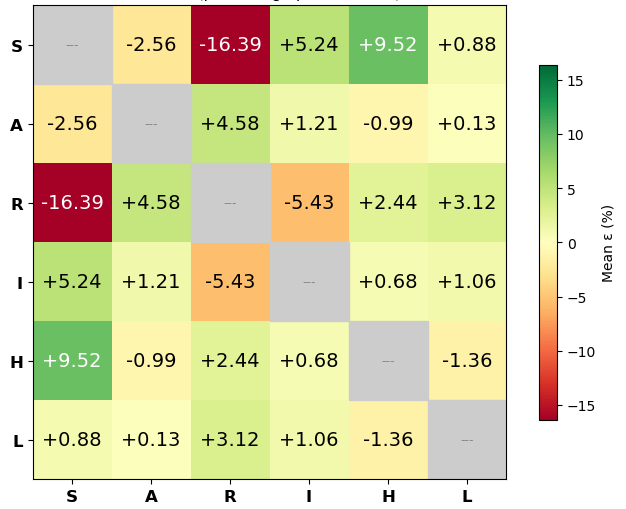}
	\caption{Mean epistasis  $\mu (\epsilon) / 1024$ (green = synergy, red = antagonism)}
	\label{epistasis_dif}
\end{figure}

\begin{table}[H]
	\centering \rowcolors{2}{white}{gray!15} \small
	\begin{tabular}{cp{4.3cm}cp{4.5cm}}
		\hline
		\textbf{C} & \textbf{Role} & \textbf{Walsh \%} & \textbf{Observation} \\ \hline
		$R$ & Dominant positive driver & 44.28 & Necessary condition; $\varnothing\!\Rightarrow\!\texttt{R}$ with support 1.00 in elite context \\
		$S$ & Cond. positive / antagonist & 7.23 & Positive alone but strong antagonism with $R$ \\
		$H$ & Consistent individual degrader & 6.67 & Negative alone; synergist in $S{\times}H$ \\
		$I$ & Consistent individual degrader
		& 5.51 & Negative alone (Lift\,$=0.80$); activates only via $A{+}R$ backbone at order~3 \\
		$L$ & Silent epistatic synergist & 0.34 & Near-zero individual effect; enters exclusively through $R{\times}I{\times}L$ and higher-order terms \\
		$A$ & Pure epistatic amplifier & $<0.1$ & No marginal effect; activates via $A{\times}R$ and $A{\times}R{\times}I$ \\
	\end{tabular}
	\caption{Role of each constraint across all analysis steps
		--- differential fitness}
\end{table}

\subsection{Linear Resistence}

This analysis reveals a design space that has grown substantially more complex across successive generator improvements. Six structural findings emerge from the convergence of marginal effects, pairwise epistasis, formal concept analysis, and LASSO regression.

\paragraph{Homogeneity remains the dominant positive driver  ---  but is weakening.}

Homogeneity ($H$) retains the highest individual lift ($2.22\times$ on \texttt{Fit\_bin})
and the broadest presence in the elite sub-space ($14/17 = 82\%$ of \texttt{Fit\_elite}
architectures). Its marginal mean is $\mu = +2{,}535$, the strongest positive signal across all constraints. However, its coefficient of variation has risen to $\mathrm{CV} = 1.41$, markedly less
stable than in runs~1--3 where $\mathrm{CV} \in [0.62,\, 0.68]$. $H$'s contribution in the LASSO model is only $3.8\%$ of explained variance  ---  a modest share in a 27-term model  ---  confirming that $H$ is increasingly a necessary but insufficient condition, whose full effect is realised only through interaction context.

\paragraph{The S--L symmetry is the defining structural feature.}

Stratification ($S$) and Locality ($L$) now exhibit near-identical marginal profiles:
$\mu_S = -2{,}062$ ($\mathrm{CV} = 1.64$, $23/31$ negative edges) and
$\mu_L = -2{,}139$ ($\mathrm{CV} = 1.39$, $25/31$ negative edges), with identical lift values of $0.53\times$ on \texttt{Fit\_bin}. Both are consistent inhibitors with comparable magnitude and directional stability. Strikingly, their pairwise epistasis achieves $\mathrm{SNR} = 0.792$  ---  the second strongest interaction in the dataset  ---  revealing a compensatory synergy: when both $S$ and $L$ are simultaneously active, their negative individual effects partially cancel. This structural regularity, absent from all previous runs, suggests that $S$ and $L$ restrict the circuit's wiring degrees of freedom through complementary mechanisms that constructively interfere under joint activation.

\paragraph{R+I is an antagonism.}

In runs~1--3, Regularity and Interleaving formed a positive synergy ($\mathrm{SNR} \in [0.35,\, 0.82]$). In run~5, this interaction reverses: $\varepsilon_{R,I} = -2.39\mathrm{k}$, $\mathrm{SNR} = 0.692$. Interleaving's cross-block wiring creates path-length variations that Regularity attempts to suppress, resulting in a structural conflict under joint activation. This inversion is the clearest evidence that the generator architecture fundamentally reshapes constraint interaction structure across versions, not merely rescaling individual contributions.

\paragraph{R+H is the dominant pairwise signal.}

For the third consecutive run, R+H achieves the highest pairwise SNR ($0.926$),
yet remains below the reliability threshold of $1.0$.
The three-way term $S \times R \times H$ is the single largest LASSO contributor
($+8.4\%$), indicating that R+H synergy is real but increasingly mediated by
higher-order context: the benefits of combining Regularity and Homogeneity are
amplified when Stratification is also active, despite $S$'s negative main effect.

\paragraph{Spectral variance decomposition.}

With 27 terms selected by LASSO and $R^2 = 0.880$, run~5 presents the most complex
model of the series. No single term exceeds $8.4\%$ of explained variance  --- 
in contrast to runs~1--3 where $H$ alone accounted for $19$--$26\%$.
The top-5 contributors together represent only $35\%$ of explained variance,
compared to $35$--$55\%$ for the top-3 terms in earlier runs.
This dispersion reflects a genuinely distributed landscape where numerous constraint
combinations make small but non-negligible contributions,
making it fundamentally harder to optimise with small populations.

\paragraph{Summary tables.}

\begin{table}[H]
	\centering \rowcolors{2}{white}{gray!15}
	\begin{tabular}{clp{3.5cm}p{5cm}} \hline
		\textbf{Order} & \textbf{Var (\%)} & \textbf{Leading terms} & \textbf{Interpretation} \\
		\hline
		1 & \textbf{42.2} & $H,\ L,\ S,\ R$ & Partially additive; dominated by $H$, $L$, $S$ in opposition \\
		2 & 24.8 & $R{\times}H$, $S{\times}R$, $R{\times}I$ & Strong pairwise structure; $R{\times}H$ alone accounts for 9.7\% \\
		3 & 24.4 & $S{\times}R{\times}H$, $S{\times}A{\times}I$, $A{\times}H{\times}L$ & Diffuse ternary epistasis of comparable weight to order~2 \\
		4 & 6.7  & $S{\times}A{\times}R{\times}H$, $A{\times}R{\times}I{\times}L$ & Residual context-specific effects \\
		5 & 1.9  & $A{\times}R{\times}I{\times}H{\times}L,\ \ldots$ & High-order interactions \\
		6 & 0.1  & $S{\times}A{\times}R{\times}I{\times}H{\times}L$ & Negligible \\
	\end{tabular}
	\caption{Spectral variance decomposition by interaction order -- linear fitness}
\end{table}

\begin{figure}[H]
	\centering
	\includegraphics[width=\textwidth]{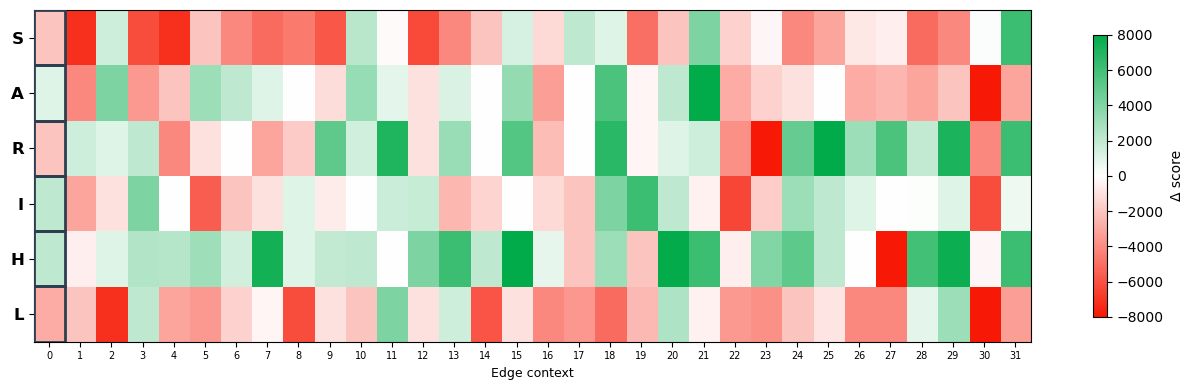}
	\caption{Sign of $\Delta C_i$ across 32 contexts (green = +, white = 0, red = --)}
	\label{heatmap_dif}
\end{figure}

\begin{figure}[H]
	\centering
	\includegraphics[width=6cm]{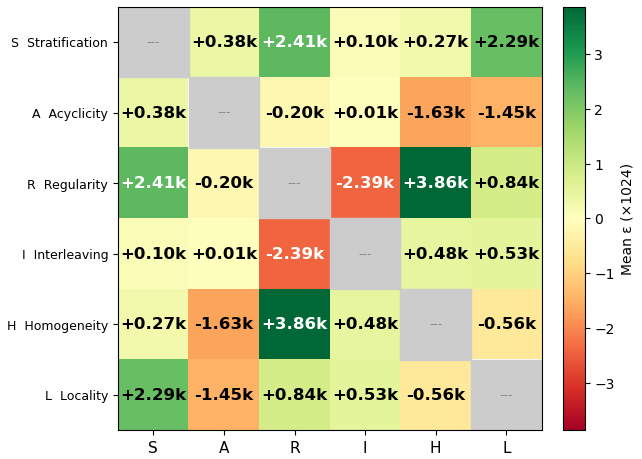}
	\caption{Mean epistasis  $\mu (\epsilon) / 1024$ (green = synergy, red = antagonism)}
	\label{epistasis_dif}
\end{figure}

Table~\ref{tab:constraint_roles} summarises the role of each constraint as
identified across the five analytical steps.

\begin{table}[H]
	\centering \rowcolors{2}{white}{gray!15}
	\begin{tabular}{clcp{4.5cm}}
		\hline
		\textbf{C} & \textbf{Role} & \textbf{Walsh \%} & \textbf{Observation} \\\hline
		$H$ & Dominant positive driver         & 15.75 & Necessary condition; activating $H$ consistently raises fitness \\
		$L$ & Consistent negative driver       & 11.57 & Symmetric with $S$; avoid unless paired with $S$ \\
		$S$ & Inhibitor / conditional synergist & 10.53 & Negative alone; recovers in $S{\times}R{\times}H$ \\
		$R$ & Positive backbone                & 4.10  & Activate with $H$; antagonistic with $I$ alone \\
		$A$ & Pure epistatic amplifier         & $<$1  & No marginal effect; activates via $A{\times}R$ and $A{\times}H$ \\
		$I$ & Neutral / context-dependent      & $<$1  & Near-zero marginal; selective use with $A{+}H$ only \\
	\end{tabular}
	\caption{Role of each constraint across all analysis steps  ---  linear fitness, run~5}
\end{table}

\subsection{Resistence to algebraic attacks}

Scores are continuous floats in $[12.0,\ 15.0]$  ---  34 distinct values, quasi-unimodal distribution with no majority plateau. Total variance is $\mathrm{Var}(f) = 0.696$ (standard deviation $0.834$). The Generic architecture (all=0) scores $12.875$, at the 22nd percentile: the unconstrained circuit is genuinely mediocre, and constraints consistently improve performance over this baseline.

\paragraph{R is the dominant constraint  ---  necessary and sufficient for the top tier.}

$R$ is the only constraint present in all 10 Fit\_elite architectures (100\%), versus 50\% in the full dataset. Its mean marginal effect is $\mu(\Delta R) = +1.30$, $\mathrm{CV} = 0.46$ (stable)  ---  the only constraint with $\mathrm{CV} < 1$. The Generic $\to$ R+Generic edge produces the largest observed jump: $+1.75$ ($12.875 \to 14.625$). $R$ alone reaches $14.625$ (Fit\_elite); $R+L$ equally reaches $14.625$.

FCA formalises this necessity: the rule $\varnothing \Rightarrow R$ (support $= 1.0$, coverage $= 10$) establishes that $R$ is a \emph{necessary condition} for membership in the top-10. In Fit\_top (score $> 13.75$), the rule $L \Rightarrow R$ (support $= 0.375$, coverage $= 12$) indicates that $L$ without $R$ always falls below the threshold. The lift of $R$ is $\mathbf{7.0}$  ---  the Fit\_top rate rises from $12.5\%$ ($R$ inactive) to $87.5\%$ ($R$ active). The Walsh decomposition confirms and quantifies: $\hat{f}(\{R\}) = -0.657$ accounts for \textbf{62.1\% of total variance}  ---  the dominant coefficient, $6\times$ larger than the second ($L$, $10.8\%$).

\paragraph{L is the main destructor  ---  avoid without R.}

$\mu(\Delta L) = -0.56$, $\mathrm{CV} = 1.09$ (moderate). $L$ is negative in 26 of 31 regular edges. The two worst architectures are H+L ($12.0$) and I+L ($12.125$). The lift of $L$ is $0.60$  ---  $L$ reduces the Fit\_top rate from $62.5\%$ to $37.5\%$.

Walsh: $\hat{f}(\{L\}) = +0.274$ (positive sign $\Rightarrow$ activating $L$ \emph{reduces} fitness), \textbf{10.8\% of variance}. $L$ is the second-largest additive effect, in direct opposition to $R$. Combined with $R$, $L$ is neutral ($R+L = 14.625$); alone, it is destructive.

\paragraph{S+A is the dominant synergy  ---  unexpected and robust.}

The pair $S+A$ has the highest epistatic SNR ($\mathrm{SNR} = 0.73$), with $\mu(\varepsilon) = +0.55$; synergy is positive in 9 of 16 contexts. In isolation, $S$ and $A$ have weak and unstable effects ($\mathrm{CV} = 2.6$ and $6.2$ respectively), but their co-activation produces a systematic gain.

LASSO: $S \times A$ is the 2nd contributor with $14.1\%$ of standardised importance (coefficient $= +0.291$). Walsh confirms: $\hat{f}(\{A,S\}) = +0.138$, $2.7\%$ of variance. $S$ appears in 40\% of Fit\_elite architectures  ---  a rehabilitation from run v2, where it was a destructor.

\paragraph{S+R is a strong antagonism.}

$\mu(\varepsilon_{S,R}) = -0.45$, $\mathrm{SNR} = 0.58$ (2nd highest), antagonism in 12/16 contexts. Despite $S$ appearing in the two best architectures (S+R+I+H and S+R+H, both score $15.0$), the simultaneous activation of $S$ and $R$ is \emph{sub-additive} on average: the combination delivers less than the sum of the two individual effects.

LASSO: the three-way term $S \times A \times R$ is the 5th contributor at $-9.2\%$.

\paragraph{H is a stable secondary driver; I is neutral.}

$\mu(\Delta H) = +0.21$, $\mathrm{CV} = 2.97$. $H$ is present in 60\% of Fit\_elite architectures; LASSO: $H$ is the 3rd term ($10.7\%$ importance); Walsh: $\hat{f}(\{H\}) = -0.112$, $1.81\%$ of variance. The interaction $H \times L$ is antagonistic (LASSO: $-7.1\%$). By contrast, $\mu(\Delta I) = -0.093$, $\mathrm{CV} = 5.66$, lift $= 1.00$; $I$ does not discriminate Fit\_top and contributes only through higher-order interactions ($S \times R \times I$ in LASSO, term 11).

\paragraph{Spectral variance decomposition.}

The \textbf{epistasis ratio is 22.6\%}  ---  the landscape is predominantly additive ($R$ dominates) but 22.6\% of variance arises from genuine interactions not capturable by a linear model. The \textbf{effective dimensionality is $2.5/63$} --- the landscape behaves as a function of 2 to 3 effective variables despite 6 nominal dimensions.

LASSO selects 13 terms ($R^2 = 0.84$) and recovers 4 of the 10 largest Walsh coefficients. Terms shared by both analyses ($R$, $L$, $S \times A$, $H$, $H \times L$) carry the strongest combined analytical support.

\paragraph{Summary tables.}

\begin{table}[H]
	\centering \rowcolors{2}{white}{gray!15}
	\begin{tabular}{clp{3cm}p{5cm}} \hline
		\textbf{Order} & \textbf{Var (\%)} & \textbf{Leading terms} & \textbf{Interpretation} \\
		\hline
		1 & \textbf{77.4} & $R,\ L,\ S,\ H,\ A,\ I$ & Near-additive landscape dominated by $R$ \\
		2 & 7.5 & $S\times A$, $R\times S$,$H\times L$, $A\times L$ & Moderate pairwise interactions \\
		3    & 5.6 & $A\times H\times L$, $I\times L\times S$, $H\times R\times S$ & Diffuse ternary epistasis \\
		4    & 3.8 & $H{\times}I{\times}L{\times}R,\ \ldots$ & Residual context-specific effects \\
		5    & 5.7 & $A{\times}H{\times}I{\times}L{\times}R,\ \ldots$ & High-order interactions \\
		6    & 0.02 & $S{\times}A{\times}R{\times}I{\times}H{\times}L$ & Negligible 
	\end{tabular}
	\caption{Spectral variance decomposition by interaction order -- Algebraic fitness}
\end{table}

\begin{figure}[H]
	\centering
	\includegraphics[width=\textwidth]{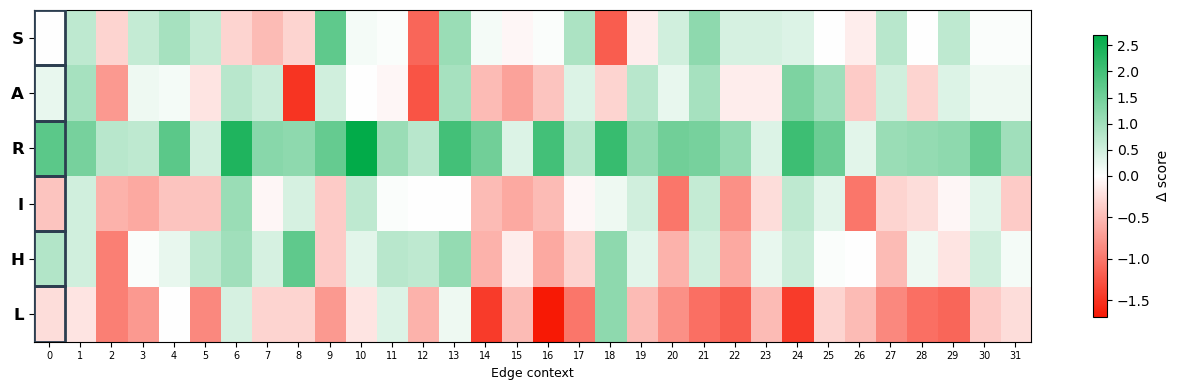}
	\caption{Sign of $\Delta C_i$ across 32 contexts (green = +, white = 0, red = --)}
	\label{heatmap_dif}
\end{figure}

\begin{figure}[H]
	\centering
	\includegraphics[width=6cm]{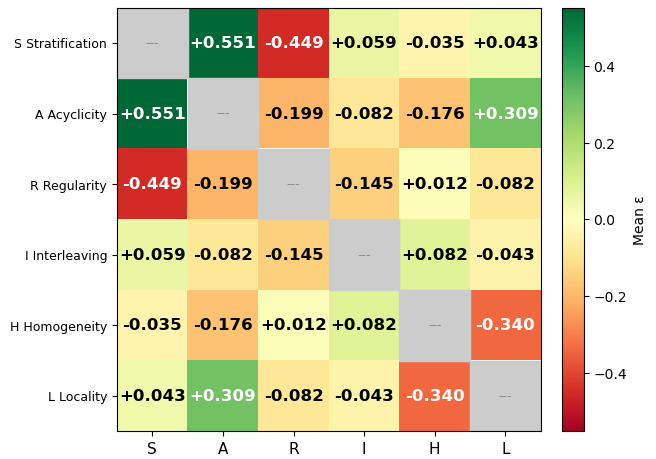}
	\caption{Mean epistasis  $\mu (\epsilon) / 1024$ (green = synergy, red = antagonism)}
	\label{epistasis_dif}
\end{figure}

\begin{table}[H]
	\centering \rowcolors{2}{white}{gray!15}
	\begin{tabular}{clcp{4cm}}
		\hline
		\textbf{C} & \textbf{Role} & \textbf{Walsh \%} & \textbf{Observation} \\\hline
		$R$ & Universal necessary condition & 62.1 & Backbone of every top architecture \\
		$H$ & Stable secondary driver      & 1.81 & Activate on $R$ backbone \\
		$S$ & Synergistic with $A$ and $R$ & 1.88 & Activate with $R$ or $A$; never alone \\
		$A$ & Synergistic with $S$          & $<1$ & Activate to trigger $S \times A$ \\
		$I$ & Neutral / unstable            & $<1$ & Selective activation only \\
		$L$ & Destructor without $R$       & 10.8 & Avoid unless $R$ is active \\
		\hline
	\end{tabular}
	\caption{Role of each constraint across all analysis steps}
\end{table}

\subsection{Top-performing architectures}

Table~\ref{multifitness} summarises the top-performing architectures for each fitness objective. No single architecture achieves elite rank in all three simultaneously: \texttt{A+R+L} is the differential champion (score $1.000$) but contains neither \texttt{H} nor \texttt{S}, which are required for linear and algebraic performance respectively. Conversely, the algebraic co-champions \texttt{S+R+I+H} and \texttt{S+R+H} (score $15.0$) activate both \texttt{S} and \texttt{R} in a high-order compensating context that degrades differential uniformity. The architectures with the strongest multi-fitness profile are those that combine \texttt{R} with \texttt{H} without \texttt{S}: \texttt{R+H} and \texttt{A+R+H} are consistently present in the upper quartile for all three objectives, with \texttt{R} providing the structural backbone, H supplying additive gain for linear and algebraic resistance, and the absence of S avoiding the dominant antagonism. Among architectures containing all three key components (\texttt{R}, \texttt{H}, \texttt{L}), \texttt{A+R+H+L} appears in the top-$3$ for differential fitness, is structurally aligned with the linear R$+$H backbone, and benefits from the $R+L$ neutrality observed for algebraic degree. It is the closest single architecture to a Pareto-efficient solution in this design space.

\begin{table}[H]
	\centering \rowcolors{2}{white}{gray!15} \small
	\begin{tabular}{llll}
		\hline
		\textbf{Arch.} & \textbf{Differential} & \textbf{Linear} & \textbf{Algebraic} \\ \hline
		\texttt{A+R+L}   & \textbf{Elite} (score $1.000$) & Moderate (no H) & Strong (R backbone) \\
		\texttt{R+L}     & \textbf{Elite} (score $0.992$) & Moderate (no H) & \textbf{Elite} ($14.625$) \\
		\texttt{A+R+H+L} & \textbf{Elite} (score $0.982$) & Strong (R+H) & Strong \\
		\texttt{R+H}     & Strong (R active) & \textbf{Dominant pair} & Strong (H secondary) \\
		\texttt{S+R+H}   & Degraded (S$\times$R) & Strong (3-way $S{\times}R{\times}H$) & \textbf{Elite} (score $15.0$) \\
		\texttt{S+H}     & Secondary ($0.820$) & Mixed (S inhibitor) & Moderate \\
		\texttt{Generic} & Weak & Weak & 22nd pct. \\
		\hline
	\end{tabular}
	\caption{Qualitative multi-fitness profile of key architectures}
	\label{multifitness}
\end{table}

\section{Discussion}   

\paragraph{Landscape complexity as a function of constraint count.}

A notable transversal finding is that multi-constraint architectures are not monotonically better. Architectures with 5--6 active constraints systematically appear in the middle tiers of the differential ranking, where the combinatorial overhead of mutually antagonistic constraints prevents the genetic algorithm from locating the structural configurations that matter. The data point to a \emph{sparse constraint hypothesis}: elite cryptographic performance in this design space is achieved not by activating many structural rules simultaneously, but by identifying and enforcing the minimal subset of non-antagonistic constraints that are co-aligned with a given fitness objective. In all three fitnesses, the best architectures use at most four constraints, and the constraint \texttt{R} is universally present.

The broader implication is epistemological. Classical cipher design implicitly treats architectural constraints as engineering specifications --- rules that guarantee desirable properties by construction, where more rules mean more guarantees. The present results suggest a different framing: in a Boolean network design space, constraints are \emph{fitness landscape shapers}, and their interactions can be as determinative as their individual effects. Activating a constraint does not simply add a property; it restructures the entire landscape of reachable architectures, potentially activating antagonisms that negate the intended benefit. Designing for cryptographic quality at the circuit level is therefore less a matter of accumulating structural guarantees and more a matter of identifying the sparse, mutually compatible constraint combinations that collectively channel the search towards high-fitness regions --- a fundamentally combinatorial and epistatic problem that classical design vocabulary was not built to address.

\paragraph{Performance of Standard Architectures.}

The two dominant paradigms of symmetric cipher design --- Feistel networks and
Substitution-Permutation Networks --- can be located precisely in the constraint
space as \texttt{A+I} and \texttt{S+A+R+I+H} respectively.
Their rankings across the three fitness objectives are summarised in
Table~\ref{ranking}.

\begin{table}[H]
	\centering \small
	\begin{tabular}{lccc}
		\hline
		\textbf{Architecture} & \textbf{Differential} & \textbf{Linear} & \textbf{Algebraic} \\ \hline
		Feistel (\texttt{A+I}) & 52/64 & 23/64 & 56/64 \\
		SPN (\texttt{S+A+R+I+H}) & 21/64 & 3/64 & 16/64 \\
		\hline
	\end{tabular}
	\caption{Ranking of standard architectures across fitness objectives}
	\label{ranking}
\end{table}

The Feistel architecture (\texttt{A+I}) performs consistently poorly, placing in the bottom third across all three objectives and in the bottom $12\%$ for differential and algebraic fitness. This is not paradoxical: Feistel's canonical strength --- the guaranteed invertibility of the round function regardless of its algebraic properties --- is a cipher-construction property, not an S-box quality criterion. The epistatic interaction between \texttt{A} and \texttt{I} is positive in both the linear and algebraic matrices, confirming that the pair is at least structurally coherent; but this synergy is too weak to compensate for the absence of any constraint linked to the dominant fitness drivers. In short, Feistel encodes \emph{reversibility} rather than \emph{cryptographic quality}: it is a cipher-engineering paradigm that scores poorly in a design space optimised for single-pass S-box strength.

The SPN architecture (\texttt{S+A+R+I+H}) tells a more nuanced and instructive story. Its linear resistance ranking (3rd out of 64) confirms that the SPN template directly encodes the structural prerequisites for spectral flatness: the \texttt{S+H} combination instantiates exactly the alternating nonlinear/linear layer structure that maximises diffusion and nonlinearity degree in a single pass, while \texttt{R} and \texttt{A} ensure structural regularity and acyclicity. Yet the same architecture falls to 21st place for differential fitness --- a drop of 18 positions that reveals a fundamental tension internal to the SPN constraint set. The culprit is precisely the (\texttt{S,R}) antagonism, identified as the strongest and most reliable epistatic signal in the differential landscape: stratification's asymmetric layer-type alternation is structurally incompatible with Regularity's uniform-depth discipline, and this conflict is activated in the SPN template. The property that makes SPN excellent for linear resistance thus actively degrades its differential performance, exposing a \emph{fitness-objective trade-off that is intrinsic to the SPN design paradigm}.

Both classical architectures diverge sharply in their relationship to \texttt{R}. The SPN template activates \texttt{R}, which provides the backbone for its dominant linear performance; but \texttt{R}'s benefit for differential fitness is precisely what the $\texttt{S,R}$ antagonism undermines. By contrast, the absence of \texttt{R} explains the poor ranking of the Feistel model for algebraic and differential fitness. Uniform propagation depth is a micro-structural property of the computational graph that has no natural formulation in the round-function vocabulary of classical cipher design, and Feistel never required it. The present results suggest that \texttt{R} is not merely one useful constraint among others, but a structural prerequisite whose explicit enforcement separates the elite architectures in this design space from all classical templates that either ignore it or activate it only to have its benefit partially cancelled by an antagonistic co-constraint.

These conclusions must be interpreted within the boundaries of the present framework. The 64-architecture design space explored here is the product of six independent binary constraints chosen for their orthogonality and structural interpretability --- a deliberately broad combinatorial canvas that includes many configurations no classical designer would ever consider. It is in this context that Feistel's mediocre ranking should be particularly appreciated. What the results do establish is that the \emph{design vocabulary} of classical cipher construction --- built around round functions, Feistel branches, and S/P alternation --- explores only a thin slice of the structural space available at the Boolean network level, and that this slice is not the one that optimises single-step cryptographic quality across all three criteria simultaneously.

\section{Conclusion and Future Work} 

Shannon's foundational intuition --- that strong ciphers require the systematic alternation of confusion and diffusion operations --- finds a precise structural expression in this framework as the Stratification constraint \texttt{S}, which enforces alternating nonlinear and linear layers. The results both vindicate and qualify this intuition. They vindicate it in the sense that \texttt{S} contributes positively to linear and algebraic resistances: the confusion-diffusion principle captures something real about the structure of high-fitness circuits.

Shannon's foundational intuition --- that strong ciphers require the systematic alternation of confusion and diffusion operations --- finds a precise structural expression in this framework as the Stratification constraint \texttt{S}, which enforces alternating nonlinear and linear layers. The results both vindicate and qualify this intuition. They vindicate it in the sense that \texttt{S} contributes positively to linear and algebraic resistances: the confusion-diffusion principle captures something real about the structure of high-fitness circuits. But the most striking finding is not that \texttt{S} matters, but that a structurally deeper constraint --- Regularity (\texttt{R}), which enforces uniform propagation depth across all circuit paths --- is the only universal necessary condition across all three fitness objectives, a principle that has no counterpart in Shannon's confusion-diffusion vocabulary. Shannon's insight was formulated at the level of functional alternation between operation types; the present results suggest that an equally fundamental structural principle operates one level below, at the topology of the computational graph itself.

The relationship between \texttt{S} and \texttt{R} is, however, not one of complementarity but of antagonism. The \texttt{S}+\texttt{R} interaction is the strongest and most reliable epistatic signal identified in this study ($\mathrm{SNR} = 1.71$, $16/16$ antagonistic contexts for differential fitness), and it is consistently negative across all three objectives unless mediated by a compensating third constraint. Stratification's asymmetric layer-type alternation is structurally incompatible with Regularity's uniform-depth discipline: the two principles pull in opposite directions at the circuit level. This antagonism is arguably the central empirical result of this work, and it carries a precise epistemological implication. Shannon's confusion-diffusion paradigm and the regularity principle identified here are not two components of a unified design theory --- they are competing structural logics, each capturing a distinct and partially irreconcilable aspect of cryptographic quality. The confusion-diffusion paradigm describes \emph{what} operations should alternate; Regularity specifies \emph{how uniformly} their effects should propagate. That these two desiderata conflict at the circuit level suggests that the design space of Boolean networks is richer, and more constrained, than either principle alone would imply.

\paragraph{Open Problems and Future Directions.}

The results presented in this work open several research directions that extend well beyond the immediate design space explored here.

\subparagraph{Multi-objective Pareto optimisation.}
The three fitness objectives studied here --- differential uniformity, linear resistance, and algebraic degree --- were analysed independently, leaving the multi-objective landscape largely unexplored. The qualitative cross-fitness profiles assembled in this work suggest that the Pareto frontier is narrow: a small number of architectures appear competitive across all three criteria simultaneously, with \texttt{R+L} and \texttt{A+R+H+L} as plausible non-dominated candidates. Characterising this frontier rigorously requires a dedicated multi-objective genetic algorithm --- NSGA-II or MOEA/D being natural candidates --- capable of maintaining a diverse population of non-dominated solutions across the full 64-architecture space. Beyond architecture selection, the multi-objective framing raises a deeper question: whether the epistatic antagonisms identified within single fitness landscapes (most notably \texttt{S+R}) manifest as structural conflicts between objectives, or whether they are objective-specific phenomena that cancel across the Pareto frontier. The answer would clarify whether a single universal constraint backbone exists for multi-criteria cryptographic optimisation, or whether any non-dominated architecture necessarily embodies a genuine design trade-off.

\subparagraph{Scalability of the constraint landscape.}
All results reported here were obtained on 16-bit Boolean networks. The extent to which the structural conclusions --- the universality of R, the \texttt{S+R} antagonism, the sparse constraint hypothesis --- generalise to larger circuit sizes is an open empirical question. Bit width is not merely a computational parameter: increasing from 16 to 32 or 64 bits changes the combinatorial structure of the differential distribution table, the Walsh spectrum density, and the effective dimensionality of the fitness landscape in ways that are not analytically tractable. It is conceivable that the R-backbone dominance weakens at larger scales, where uniform path depth becomes harder to enforce without imposing prohibitive structural rigidity, or conversely that it strengthens, as the penalty for non-uniform differential propagation grows with circuit width. Systematically replicating the six-constraint analysis at 32 and 64 bits --- with appropriately scaled genetic algorithm budgets and GPU-accelerated fitness evaluation --- would establish whether the present findings constitute a universal principle of Boolean circuit cryptographic design or a size-specific regularity.

\subparagraph{Extension to iterated circuits.}
The single-pass evaluation model adopted in this work isolates the intrinsic structural quality of each architecture but deliberately abstracts away the iterative round structure that defines practical block ciphers. Extending the framework to multi-round SBNs --- where the same constrained Boolean network is applied repeatedly to its own output, with round keys mixed at each step --- would bridge the gap between S-box optimisation and full cipher design. This extension is non-trivial in several respects. Fitness evaluation cost grows linearly with the number of rounds, making differential resistance assessment particularly expensive. More fundamentally, the constraint interaction structure may change qualitatively under iteration: antagonisms that are damaging in a single pass may become irrelevant or even beneficial when their effects compound across rounds, as the Shannon confusion-diffusion argument implicitly assumes. The central question is whether the sparse R-backbone architectures that dominate single-pass performance retain their advantage in an iterated setting, or whether the additional mixing provided by multiple rounds effectively neutralises architectural differences --- in which case classical design intuition would be recovered as a limiting case of the present framework.

\subparagraph{Cryptographic Robustness / Hardware Implementation Cost Trade-off.}
The fitness objectives studied --- differential uniformity, linear resistance, and algebraic degree --- are purely cryptographic criteria evaluated on the Boolean network's input-output behaviour. They say nothing about the cost of implementing the winning architectures in hardware. Yet the constraints that drive elite cryptographic performance are precisely those that tend to increase implementation complexity. The architectures that dominate the differential fitness ranking (\texttt{A+R+L}, \texttt{R+L}) are, from a hardware perspective, unusually constrained graphs whose silicon area and critical-path delay have not been evaluated. A natural and practically important extension of this work is therefore to augment the fitness framework with a hardware cost model --- gate count, wire length, or circuit depth under a realistic cell library --- and to study the Pareto frontier between cryptographic quality and implementation efficiency across the 64 architectures. It is an open question whether the elite cryptographic architectures identified here remain competitive once hardware cost is factored in, or whether the sparsely constrained R-backbone configurations impose an implementation overhead that outweighs their cryptographic advantage over classical SPN templates.


\bibliographystyle{plain}
\bibliography{sbn.bib}
	
\appendix

\section{Proof of proposition \ref{proof} by witnesses}

We write the predicate vector of an enriched SBN as $\nu(\mathcal C)=(S,A,R,I,H,L)\in\mathbb{B}^6$. For each coordinate, we exhibit two witnesses differing only on that coordinate.
	
\paragraph{1. Toggling \(S\).}
	Take $V=\{a,b,c\},\qquad E=\{a\to b,\ b\to c\}$, with layer map $\ell(a)=1$, $\ell(b)=2$, $\ell(c)=3$, boundary sets $V_{\mathrm{in}}=\{a\}$, $V_{\mathrm{out}}=\{c\}$, one single block, and a metric such that all used edges are local.
	
	For $\mathcal C_S^+$, choose an alternating type map $\tau(1)=\mathsf S$, $\tau(2)=\mathsf P$, $\tau(3)=\mathsf S$, and choose local rules consistent with these layer types.
	
	For \(\mathcal C_S^-\), keep the same graph, same layer map, same blocks, same metric, but choose a non-alternating typing, for instance $\tau(1)=\mathsf S$, $\tau(2)=\mathsf S$, $\tau(3)=\mathsf P$, again with local rules consistent with these assigned layer types.
	
	Then $S(\mathcal C_S^+)=1$, $S(\mathcal C_S^-)=0$, while all other predicates are unchanged: $A=1$, $R=1$, $I=0$,$H=1,\quad L=1$. Hence
	\[
	\nu(\mathcal C_S^+)=(1,1,1,0,1,1),\qquad
	\nu(\mathcal C_S^-)=(0,1,1,0,1,1).
	\]

\paragraph{2. Toggling \(A\).}
	Take \(V=\{a,b\}\), with $\ell(a)=1$, $\ell(b)=2$, $\tau(1)=\mathsf S$, $\tau(2)=\mathsf P$, $V_{\mathrm{in}}=\{a\}$, $V_{\mathrm{out}}=\{b\}$, one block, and local metric.
	
	Let $E^+=\{a\to b\}$, $E^-=\{a\to b,\ b\to a\}$.
	Then $T=1$ in both cases, since the alternating layer typing is unchanged and does not constrain edge orientation. However, $A(\mathcal C_A^+)=1$, $A(\mathcal C_A^-)=0$.
	All other predicates remain unchanged: $R=1$, $I=0$, $H=1$,$L=1$. Thus
	\[
	\nu(\mathcal C_A^+)=(1,1,1,0,1,1),\qquad
	\nu(\mathcal C_A^-)=(1,0,1,0,1,1).
	\]

\paragraph{3. Toggling \(R\).}
	For \(R=1\), take vertices $a\in V_1$, $b,d\in V_2$, $c,e\in V_3$, with $\tau(1)=\mathsf S$, $\tau(2)=\mathsf P$, $\tau(3)=\mathsf S$, $V_{\mathrm{in}}=\{a\}$, $V_{\mathrm{out}}=\{c,e\}$, one block, homogeneous local rules inside each layer, and local metric. Let $E^+=\{a\to b,\ b\to c,\ a\to d,\ d\to e\}$. Then $\tau_G(a,c)=\tau_G(a,e)=2$, so $R=1$. Also $S=1$, $A=1$, $I=0$, $H=1$, $L=1$.
	
	For $R=0$, take vertices $a\in V_1$, $b,d\in V_2$, $e\in V_3$, with the same alternating typing $\tau(1)=\mathsf S$, $\tau(2)=\mathsf P$, $\tau(3)=\mathsf S$, $V_{\mathrm{in}}=\{a\}$, $V_{\mathrm{out}}=\{d,e\}$, and $E^-=\{a\to d,\ a\to b,\ b\to e\}$. Then $\tau_G(a,d)=1$, $\tau_G(a,e)=2$, hence \(R=0\), while still $S=1$, $A=1$, $I=0$, $H=1$, $L=1$. Therefore
	\[
	\nu(\mathcal C_R^+)=(1,1,1,0,1,1),\qquad
	\nu(\mathcal C_R^-)=(1,1,0,0,1,1).
	\]

\paragraph{4. Toggling \(I\).}
	Take the same graph \(a\to b\to c\) with $\ell(a)=1$, $\ell(b)=2$, $\ell(c)=3$, $\tau(1)=\mathsf S$, $\tau(2)=\mathsf P$, $\tau(3)=\mathsf S$, $V_{\mathrm{in}}=\{a\}$, $V_{\mathrm{out}}=\{c\}$, and local metric.
	
	For \(I=0\), choose one block: $b(a)=b(b)=b(c)=1$. For \(I=1\), choose for instance $b(a)=1$, $b(b)=2$, $b(c)=2$. Then the edge $a\to b$ is cross-block, so $I=1$. In both cases, $S=1$, $A=1$, $R=1$, $H=1$, $L=1$. Thus
	\[
	\nu(\mathcal C_I^-)=(1,1,1,0,1,1),\qquad
	\nu(\mathcal C_I^+)=(1,1,1,1,1,1).
	\]

\paragraph{5. Toggling \(H\).}
	Take $a\in V_1,\qquad b,c\in V_2,\qquad d\in V_3$, with graph $E=\{a\to b,\ a\to c,\ b\to d,\ c\to d\}$, one block, $\tau(1)=\mathsf S$,$\tau(2)=\mathsf P$, $\tau(3)=\mathsf S$, $V_{\mathrm{in}}=\{a\}$, $V_{\mathrm{out}}=\{d\}$, and local metric. Then $S=1$, $A=1$, $R=1$,$I=0$, $L=1$.
	
	For \(H=1\), choose identical affine unary rules on \(b\) and \(c\), for instance $f_b(x)=x$,$f_c(x)=x$. For \(H=0\), choose non-equivalent affine unary rules, e.g. $f_b(x)=x$,$f_c(x)=1+x$. These two rules are both of type \(\mathsf P\), so \(S\) remains equal to \(1\), but they are not equivalent up to input permutation. Hence
	\[
	\nu(\mathcal C_H^+)=(1,1,1,0,1,1),\qquad
	\nu(\mathcal C_H^-)=(1,1,1,0,0,1).
	\]

\paragraph{6. Toggling \(L\).}
	Take again the graph \(a\to b\to c\), with one block and singleton layers, so that $\ell(a)=1$, $\ell(b)=2$, $\ell(c)=3$, $\tau(1)=\mathsf S$, $\tau(2)=\mathsf P$, $\tau(3)=\mathsf S$, and therefore $S=1$, $A=1$,$R=1$, $I=0$, $H=1$. Keep the graph, block map, layer map, type map, and local rules fixed. For \(L=1\), choose a metric \(\rho^+\) and radius \(\delta_0\) such that $\rho^+(a,b)\le \delta_0$, $\rho^+(b,c)\le \delta_0$.
	For \(L=0\), choose another metric \(\rho^-\) with $\rho^-(a,b)>\delta_0$, $\rho^-(b,c)\le \delta_0$. Then $L(\mathcal C_L^+)=1$, $L(\mathcal C_L^-)=0$, while the other five predicates are unchanged:
	\[
	\nu(\mathcal C_L^+)=(1,1,1,0,1,1),\qquad
	\nu(\mathcal C_L^-)=(1,1,1,0,1,0).
	\]
	
	\medskip
	Collecting the six witness pairs proves the claim. \hfill $\square$


\section{Reward and Fitness Functions}

Let \(F : \mathbb{F}_2^n \rightarrow \mathbb{F}_2^m\) denote the Boolean transformation implemented by a BC. In the evolutionary framework, \emph{fitness functions} provide global adversarial evaluation criteria, whereas \emph{reward functions} provide local structural signals used to guide mutation and exploration.

\subsection{Fitness Functions}

We evaluate the cryptographic quality of \(F\) through three standard criteria related to differential, linear, and algebraic cryptanalysis. Each criterion defines a fitness function to be maximized.

\paragraph{Differential Fitness (F1)}

For \(\alpha \in \mathbb{F}_2^n\) and \(\beta \in \mathbb{F}_2^m\), the differential distribution table is $\mathrm{DDT}_F(\alpha,\beta)=
\left|\left\{x \in \mathbb{F}_2^n:F(x)\oplus F(x\oplus \alpha)=\beta\right\}\right|$.

The maximum differential probability is
\[
\mathrm{DP}_{\max}(F)
=
\max_{\alpha\neq 0,\;\beta}
\frac{\mathrm{DDT}_F(\alpha,\beta)}{2^n}.
\]

A cryptographically strong function is expected to satisfy \(\mathrm{DP}_{\max}(F)\leq \varepsilon_D\) for some small threshold \(\varepsilon_D\). We define $\mathrm{Fit}_{\mathrm{diff}}(F)=-\log_2\!\big(\mathrm{DP}_{\max}(F)\big)$.

Maximizing \(\mathrm{Fit}_{\mathrm{diff}}\) amounts to minimizing the success probability of the best differential characteristic.

\paragraph{Linear Fitness (F2)}

For \(a \in \mathbb{F}_2^n\) and \(b \in \mathbb{F}_2^m\), define the correlation
\[
C_F(a,b)
=
\frac{1}{2^n}
\sum_{x\in\mathbb{F}_2^n}
(-1)^{a\cdot x \oplus b\cdot F(x)}.
\]

The maximum linear bias is $\mathrm{LP}_{\max}(F) = \max_{a\neq 0,\;b\neq 0}|C_F(a,b)|$.

Equivalently, in Walsh form,
\[
\max_{a\neq 0,\;b\neq 0}
\left|
\sum_{x\in\mathbb{F}_2^n}
(-1)^{a\cdot x \oplus b\cdot F(x)}
\right|
\leq \varepsilon_L 2^n,
\]
for some small threshold \(\varepsilon_L\).

We define $\mathrm{Fit}_{\mathrm{lin}}(F) = -\log_2\!\big(\mathrm{LP}_{\max}(F)\big)$.

Maximizing \(\mathrm{Fit}_{\mathrm{lin}}\) reduces the strength of the best linear approximation.

\paragraph{Algebraic Fitness (F3)}

Algebraic cryptanalysis exploits low-degree polynomial relations between input and output variables. A strong Boolean transformation should therefore exhibit high algebraic complexity.

Let \(F_i\) denote the \(i\)-th coordinate function of \(F\), written in Algebraic Normal Form (ANF):
\[
F_i(x)
=
\bigoplus_{u\in\mathbb{F}_2^n}
a_{i,u}\,x^u.
\]

The algebraic degree of \(F\) is $\deg(F)=\max_{1\leq i\leq m}\deg(F_i)$.

To capture the growth of algebraic complexity under iteration, we consider
\[
\deg_k(F)
=
\deg\!\big(F^{(k)}\big).
\]

The algebraic fitness is then defined by $\mathrm{Fit}_{\mathrm{alg}}(F)=\deg_k(F)$.

Maximizing \(\mathrm{Fit}_{\mathrm{alg}}\) favors transformations whose algebraic complexity grows rapidly under composition, providing a practical proxy for resistance to algebraic attacks.

\subsection{Reward Functions}

Fitness functions are global and adversarial, but often expensive and locally flat. We therefore complement them with reward functions based on structural variations between a circuit \(F\) and a mutated circuit \(F'\). Each reward is defined as a differential indicator \(R_i(F\rightarrow F')\).

\paragraph{Effective Algebraic Degree Reward (R1)}
Let $\deg_k(F)=\max_{1\leq i\leq m}\deg\!\big(F_i^{(k)}\big)$.
We define $R_1(F\rightarrow F')=\deg_k(F')-\deg_k(F)$.

R1 promotes mutations increasing multi-round algebraic complexity.

\paragraph{Algebraic Normal Form Entropy (R2)}

For each coordinate \(F_i\), let $M_i=\{u\in\mathbb{F}_2^n : a_{i,u}=1\}$ be the set of monomials appearing in its ANF. Define $H_{\mathrm{ANF}}(F)=\frac{1}{m}\sum_{i=1}^{m}\frac{\log_2(|M_i|+1)}{n}$.

Then $R_2(F\rightarrow F') = H_{\mathrm{ANF}}(F')-H_{\mathrm{ANF}}(F)$.

This reward promotes dense and less factorizable ANF representations.

\paragraph{Walsh Spectrum Flattening (R3)}

For each coordinate \(F_i\), define the Walsh transform
$W_{F_i}(a)=\sum_{x\in\mathbb{F}_2^n}(-1)^{F_i(x)\oplus a\cdot x}$.

Let $S(F)=\frac{1}{m}\sum_{i=1}^{m}\mathrm{Var}_{a\neq 0}\!\big(W_{F_i}(a)\big)$.

Then $R_3(F\rightarrow F') = S(F)-S(F')$.

Reducing spectral dispersion discourages dominant linear correlations.

\paragraph{Local Differential Uniformization (R4)}

Define $\mathrm{DDT}_F(\alpha,\beta)=\left|\{x : F(x)\oplus F(x\oplus\alpha)=\beta\}\right|$,
and $D(F)=\mathrm{Var}_{\alpha\neq 0,\beta}\left(\frac{\mathrm{DDT}_F(\alpha,\beta)}{2^n}\right)$.

Then $R_4(F\rightarrow F') = D(F)-D(F')$.

This reward penalizes uneven differential behavior and dominant differential transitions.

\paragraph{Dependency Graph Expansion (R5)}

Let \(M(F)\) be the dependency matrix
\[
M_{i,j}(F)=
\begin{cases}
	1 & \text{if } F_i \text{ depends on } x_j,\\
	0 & \text{otherwise}.
\end{cases}
\]

Let $\rho(F)=\mathrm{rank}(M(F))$.

Then $R_5(F\rightarrow F') = \rho(F')-\rho(F)$.

This reward promotes stronger input-output entanglement.

\paragraph{Symmetry Breaking Reward (R6)}

Let \(\mathrm{Aut}(F)\) denote the automorphism group of the circuit structure, and define $\sigma(F)=|\mathrm{Aut}(F)|$.

Then $R_6(F\rightarrow F') = \sigma(F)-\sigma(F')$.

This reward penalizes structural symmetries and repeated exploitable patterns.

\section{SBN Explorer}

The software \emph{SBN Explorer} is available on GitHub.

\subsection{Genetic Algorithm Configuration}

\paragraph{Formal classification}
The implemented algorithm is a \emph{Generational Genetic Algorithm with $(\mu+\lambda)$ Elitism and Truncation Selection}, also described in the literature as a $(\mu+\lambda)$-ES with Poisson mutation. It differs from a classical genetic algorithm in the absence of a crossover operator: reproduction is entirely asexual, which brings it closer to Evolution Strategies in the sense of Schwefel \cite{schwefel1995}. This design choice is motivated by the structure of the genotype — a Boolean gate network whose cryptographic properties depend on non-contiguous distributed patterns across the circuit graph, which crossover would likely disrupt. The algorithm can therefore be characterised as an ES/GA hybrid optimised for structured search spaces.

\paragraph{Population Size}
The default population size is set to 20 individuals. This deliberately compact value reflects a constraint of the experimental design: since the algorithm is run once per architecture across all 64 combinations, the computational budget is distributed across the full design space rather than concentrated on a single run. With a population of 20, the elitist component retains the 4 best individuals ($\lfloor 20/5 \rfloor = 4$, i.e.\@ 20\%) as parents for the next generation.

\paragraph{Number of generations}
The default number of generations is 50. Combined with a population of 20, this yields 1,000 fitness evaluations per architecture, or 64,000 evaluations in total for the full experiment. Given that a differential resistance evaluation requires approximately 12 seconds on GPU, a complete run over all 64 architectures takes roughly 200 hours.

\paragraph{Mutation rate}

Mutation operates at two independent levels. At the individual level, each call to \texttt{mutate\_sbn} draws the number of successive mutations from a Poisson distribution with parameter $\lambda$: the number of applied mutations is $k \sim \mathrm{Poisson}(\lambda)$, with a floor of 1. Within each mutation, two channels are applied independently with probability 0.15 per gate: an \emph{operation mutation} that replaces the gate's Boolean function, and a \emph{wiring mutation} that rewires one of its inputs. The default value is $\lambda = 3.0$.

\begin{table}[H]
	\centering
	\begin{tabular}{ccc}
		\textbf{Mutation rate $\lambda$} & \textbf{Mean mutations} & \textbf{Genome modification} \\ \hline
		2.0 & 2.0 & $\sim$3\% \\
		3.0 & 3.0 & $\sim$5\% \\
		4.0 & 4.0 & $\sim$6\% \\
	\end{tabular}
	\caption{Impact of mutation rate on genome modification. An SBN has 16 internal gates with, on average, 4 connections each, giving a genome of approximately 64 genes.}
\end{table}

\paragraph{Selection and reproduction}
The algorithm follows a generational replacement scheme with strong elitism, classifiable as a $(\mu + \lambda)$-evolution strategy with truncation selection. Each generation, the 20\% best-scoring individuals are selected as parents (truncation selection). The next population is formed by copying the parents intact (elitism), then filling the remaining slots by repeatedly drawing a random parent and applying a Poisson-distributed chain of mutations. No crossover operator is used: recombination between two parents risks destroying functional structures that may be distributed non-contiguously across the circuit graph. The asexual strategy, combined with the Poisson number of mutations, naturally balances local exploitation (few mutations) and exploratory jumps (many mutations) within the same generation.

\paragraph{Exploration vs exploitation trade-off}
The Poisson distribution over mutation counts produces a heterogeneous offspring population: some individuals undergo minor perturbations (1--2 mutations, exploiting the current solution) while others undergo significant rewirings (5 or more mutations, exploring distant regions of the search space). This implicit diversity mechanism partly compensates for the absence of crossover and the pressure of truncation selection, which would otherwise cause premature convergence.

\subsection{Software Configuration}
The experiment is implemented as a Jupyter notebook (\texttt{sbn\_generator\_v1.ipynb}) which orchestrates the full pipeline: loading the SBN module, configuring the GPU accelerator, selecting the fitness function and architecture constraints, and iterating the genetic algorithm over all 64 architecture combinations. Results are saved as a CSV file indexed by architecture rank. The core SBN library (\texttt{sbn\_module\_v1.py}) handles circuit generation, constraint enforcement, mutation, and fitness evaluation.

\subsection{Harware Configuration}
	
Experiments were performed on a GPU instance equipped with a NVIDIA A10 GPU (shared), an AMD EPYC 7513 CPU (single shared core), and 4 GB RAM.\\
\textbf{Warning}: This code does NOT support CPU-only execution.
	
\end{document}